\PassOptionsToPackage{dvipsnames}{xcolor}
\documentclass[11pt]{article}

\usepackage[margin=1in]{geometry}
\usepackage[T1]{fontenc}
\usepackage{lmodern}
\usepackage{amsmath,amssymb,amsthm,mathtools,booktabs,bm}
\usepackage{xcolor}
\usepackage[linesnumbered,ruled,vlined]{algorithm2e}
\usepackage{nicematrix}
\NiceMatrixOptions{code-for-first-row = \scriptstyle\color{gray}, code-for-first-col = ~\scriptstyle\color{gray}}

\usepackage[authoryear]{natbib}
\setcitestyle{authoryear}
\usepackage[colorlinks=true,linkcolor=blue,citecolor=blue,urlcolor=blue]{hyperref}
\usepackage{cleveref}
\usepackage{enumerate}
\usepackage{needspace}

\newtheorem{theorem}{Theorem}[section]
\newtheorem{definition}{Definition}[section]
\newtheorem{proposition}{Proposition}[section]
\newtheorem{lemma}{Lemma}[section]
\newtheorem{example}{Example}[section]
\newtheorem{corollary}{Corollary}[section]

\makeatletter
\@for\thm@env:=definition,proposition,lemma,example,corollary,remark\do{%
  \expandafter\let\csname c@\thm@env\endcsname\c@theorem}
\makeatother

\renewcommand{\mid}{:}
\renewcommand{\epsilon}{\varepsilon}
\newcommand{\Chf}{\widetilde{C}} 
\newcommand{\Chcd}{\mathcal{C}} 
\newcommand{\Chcf}{\overline{\mathcal{C}}} 
\newcommand{\cM}{\mathcal{M}}

\newcommand{\ZZ}{\mathbb{Z}}
\newcommand{\RR}{\mathbb{R}}
\newcommand{\RRbar}{\underline{\mathbb{R}}}

\DeclareMathOperator{\dom}{dom}
\DeclareMathOperator{\conv}{conv}
\DeclareMathOperator{\supp}{supp}
\DeclareMathOperator{\argmax}{arg\,max}
\DeclareMathOperator{\argmin}{arg\,min}

\title{Stable and Fair Random Allocations in a Two-Sided Discrete-Concave Market\thanks{%
A one-page abstract of this paper is to appear in the Proceedings of the 27th ACM Conference on Economics and Computation (EC 2026).
We are especially grateful to Keisuke Bando for his detailed comments and suggestions.
We also thank Fuhito Kojima, Kazuo Murota, Akiyoshi Shioura, Akihisa Tamura, M.~Utku \"{U}nver, and Koji Yokote, as well as seminar participants at Keio University and participants in the Asian Game Theory Meeting 2026, the ERATO weekly meeting, and EC 2026, for their helpful comments.
This work was partially supported by JSPS KAKENHI Grant Numbers JP21H04979, JP23K12443, and JP25K00137; and JST ERATO Grant Number JPMJER2301.
We used OpenAI's GPT models and Anthropic's Claude models to assist with part of the proof of \Cref{thm:choice-identification} and with drafting the manuscript.
The authors take full responsibility for the content of this paper, including all mathematical claims and proofs.}}
\author{
Kenzo Imamura\\
The University of Tokyo
\and
Yasushi Kawase\\
Chuo University
}
\date{}

\begin{document}
\maketitle

\begin{abstract}
We study random allocations in two-sided many-to-many matching markets with ties, where random tie-breaking can violate ex ante stability and fairness.
We show that, when valuations are discrete concave (M$^\natural$-concave), stable and fair fractional allocations always exist and form a distributive lattice under the induced preference orders.
Every such allocation can be implemented as a lottery over stable deterministic allocations that simultaneously gives every agent the highest expected utility compatible with her fractional bundle.
Since cardinal utilities are difficult to elicit, we then ask what ordinal preferences over deterministic outcomes can identify.
The set of stable and fair fractional allocations is the same for all M$^\natural$-concave valuations consistent with the same ordinal preferences, although utility-preserving lotteries may differ across them.
No such dependence arises for additive valuations under matroid constraints, where every implementing lottery is ex post stable and utility-preserving under every consistent cardinal representation.

\medskip
\noindent\textsc{Keywords:} Two-sided matching, random allocation, ex ante stability, M$^\natural$-concavity, gross substitutes, ordinal preferences.
\end{abstract}

\section{Introduction}\label{sec:introduction}

Indifferences are common in matching markets because institutions often rank applicants using coarse priority categories, and agents may view several potential partners as equally desirable. They arise in school choice \citep{abdulkadirouglu2003school}, hospital--resident matching \citep{roth1984evolution}, and labor markets \citep{cowgill2025stable}. A standard approach is to break ties at random and then apply a deterministic stable matching mechanism. This approach generates a lottery over stable matchings, but it does not directly control the allocation before the random draw. As a result, the induced lottery may have undesirable ex ante properties. First, it may violate ex ante stability: a small transfer of probability can make both sides of a match better off. Second, it may violate ex ante fairness: agents who are tied in priority may not be treated equally with respect to assignment probabilities.

We study \emph{doubly-strong ex ante stability}, which combines ex ante stability, no ex ante discrimination on either side, and neutrality with respect to mutual indifferences \citep{kesten2015theory,cookson2025fairly}. Ex ante stability requires individual rationality and rules out a pair whose members can both strictly benefit from a small reallocation of assignment probability toward each other. No ex ante discrimination on either side rules out a reallocation toward an underassigned partner that leaves one side indifferent while making the other side strictly better off. Neutrality with respect to mutual indifferences rules out asymmetric assignment probabilities when the corresponding reallocation leaves both sides indifferent.

We focus on a two-sided
discrete-concave market: each agent has an \emph{M$^\natural$-concave valuation}, which imposes a form of substitutability while allowing nonadditive values for bundles.
When each pair can be matched at most once, M$^\natural$-concavity is equivalent to the \emph{gross substitutes} condition in matching with transferable utility \citep{kelso1982job,fujishige2003note}.\footnote{With multiple units, M$^\natural$-concavity is equivalent to the \emph{strong substitutes} condition, which is more restrictive than gross substitutes \citep{milgrom2009substitute,shioura2015gross}.} This condition is necessary (in a maximal domain sense) and sufficient for the existence of competitive equilibria \citep{gul1999walrasian,yang2017maximal}. 
In matching without transfers, when an M$^\natural$-concave valuation selects a unique optimal bundle from each availability vector, the induced choice function satisfies \emph{substitutability} and \emph{the law of aggregate demand} \citep{EguchiA2003Generalized,MurotaYokoi2015}.
In many-to-one matching with contracts, these conditions guarantee that the worker-proposing deferred acceptance mechanism is stable and strategy-proof for workers \citep{hatfield2005matching}. The class of M$^\natural$-concave valuations accommodates choice rules used in practice: it can represent institutional choice rules with type-specific quotas \citep{abdulkadirouglu2003school}, type-specific reserves \citep{hafalir2013effective,ehlers2014school}, and overlapping reserves \citep{aygun2021college,sonmez2022affirmative}.
In matching with constraints, the associated class of M$^\natural$-convex feasible sets has been used to design mechanisms under distributional constraints \citep{kojima2018designing,IK2025,BIK2025}.

Establishing doubly-strong ex ante stability in our general setting requires new arguments. 
In the responsive models of \citet{kesten2015theory} and \citet{cookson2025fairly}, stability and fairness are defined directly on the matrix of marginal assignment probabilities, and which partners an agent receives together in each realization plays no role.
A Birkhoff--von Neumann-type stable-decomposition result then implements such an allocation as a lottery over stable matchings.
With nonresponsive preferences, by contrast, how partners are grouped can affect ex ante utility.

A school that values diversity provides a simple illustration.
Consider a school with two seats and three students $a,b,c$.
Students $a$ and $b$ have the same type, while $c$ has a different type.
The school strictly prefers any two-student class to any singleton and, among two-student classes, strictly prefers a class containing both types.
It is indifferent among singletons and strictly prefers each to the empty class.
Formally,
\[
\{a,c\}\sim\{b,c\}\succ\{a,b\}\succ\{a\}\sim\{b\}\sim\{c\}\succ\varnothing.
\]
Now compare two lotteries.
The first assigns probability $1/2$ to each of $\{a,c\}$ and $\{b\}$, while the second assigns probability $1/2$ to each of $\{a,b\}$ and $\{c\}$.
Both give each student an admission probability of $1/2$.
Yet the school strictly prefers the first lottery, whose two-student class is strictly better and whose singleton is equally good.
Thus, admission probabilities alone do not determine the school's preferences over lotteries.

This example raises two questions, one at the individual level and the other at the market level.
At the individual level, how should an agent evaluate a fractional outcome?
We evaluate lotteries by expected utility.
As the example illustrates, lotteries inducing the same fractional outcome can yield different expected utilities.
To formulate ex ante stability and fairness directly in terms of fractional allocations, we must therefore specify which of these attainable utility levels to associate with each fractional outcome.
At the market level, can the utility levels assigned in this way be realized simultaneously?
Specifically, does there exist a single lottery over feasible deterministic allocations that (i) implements the given fractional allocation and (ii) delivers to each agent the utility assigned to her fractional outcome?

We first prove the existence of doubly-strong ex ante stable fractional allocations and show that they form a distributive lattice under the induced preference orders (\Cref{thm:cardinal-main}).
To establish these results, we extend each agent's valuation from deterministic to fractional bundles using its \emph{concave closure}, which assigns to each fractional bundle the maximum expected utility over lotteries consistent with it.
We then connect our model to the general choice-function framework of \citet{AlkanGale2003}, which encompasses fractional matching.
To do so, we show that, under M$^\natural$-concavity, the set of utility-maximizing fractional bundles contains a unique \emph{Lorenz-dominant} fractional bundle, using the results of \citet{Fujishige1980lex}.
We use this fractional bundle to resolve indifferences in the fractional choice correspondences, yielding single-valued choice functions.
We show that these choice functions satisfy the conditions of Alkan and Gale's existence and lattice theorems.
Crucially, Alkan and Gale's stability under these choice functions is equivalent to doubly-strong ex ante stability in our model (\Cref{thm:ag-ds-equivalence}).

We next turn to lottery implementation of these fractional allocations.
We show that every feasible fractional allocation admits a single lottery over feasible deterministic allocations that induces it and simultaneously delivers to each agent the utility assigned to her fractional bundle (\Cref{prop:joint-implementability}).
The challenge is that these utility levels are defined by maximizing expected utility separately for each agent, whereas implementation requires a common lottery.
Our proof combines two ideas.
First, the concave-closure representation allows us to use supporting hyperplanes to identify suitable deterministic bundles.
Any lottery over these bundles consistent with the agent's fractional bundle then delivers her assigned utility in expectation.
Second, M$^\natural$-concavity ensures that these bundle sets are M$^\natural$-convex.
We can therefore apply an intersection theorem---a generalization of the Birkhoff--von Neumann theorem---to obtain a common lottery consistent with both sides of the market \citep{edmonds1970submodular,Frank1984c}.
Moreover, if the fractional allocation is ex ante stable, every such utility-preserving lottery is supported entirely on stable deterministic allocations (\Cref{thm:ex-ante-to-ex-post}).

These results are stated in terms of agents' cardinal valuations, which may be difficult to elicit accurately in practice \citep{budish2022can}.
We therefore consider ordinal reports over deterministic bundles and ask what can be achieved using these reports while retaining the underlying cardinal model.
We show that all M$^\natural$-concave valuations consistent with the same reported weak orders induce the same fractional choice correspondences and hence the same set of doubly-strong ex ante stable fractional allocations (\Cref{thm:choice-identification} and \Cref{cor:ordinal-invariance}).
Thus, ordinal reports suffice to identify these allocations without eliciting the agents' cardinal valuations.

Ordinal information does not, however, suffice in general to select a utility-preserving lottery.
Although ordinal preferences identify the lotteries over stable deterministic allocations that implement a given fractional allocation, they need not identify which of them preserve every agent's utility.
Indeed, we construct a doubly-strong ex ante stable allocation for which no single implementing lottery preserves every agent's utility under all compatible M$^\natural$-concave valuations (\Cref{ex:ordinal-ambiguous-order}).
Selecting a utility-preserving lottery may therefore require additional cardinal information, even though identifying the fractional allocation itself does not.

We next restrict agents' true valuations to a \emph{matroid-responsive} subclass of M$^\natural$-concave valuations, in which valuations are additive over feasible bundles and feasibility is governed by a \emph{matroid}.
Within this subclass, known matroid constraints and weak rankings over individual partners and the outside option suffice to identify stable and fair fractional allocations and select utility-preserving stable lotteries.
Additivity ensures that every feasible lottery inducing the same fractional allocation delivers the same expected utility; consequently, every lottery implementing a stable and fair fractional allocation is supported on stable deterministic allocations (\Cref{cor:matroid-responsive-implementation}).
This class encompasses the models of \citet{kesten2015theory} and \citet{cookson2025fairly} and accommodates type-specific quotas \citep{abdulkadirouglu2003school}.

\subsection{Related Work}\label{subsec:related_work}
Our paper is most closely related to \citet{kesten2015theory} and \citet{cookson2025fairly} on stable and fair random allocations.
\citet{kesten2015theory} introduce ex ante stability and fairness notions for random allocations, propose the fractional deferred acceptance mechanism, and prove existence in many-to-one matching with responsive choice and indifferences on one side.
\citet{cookson2025fairly} introduce doubly-strong ex ante stability in one-to-one matching with indifferences on both sides.
We extend their analysis to nonresponsive bundle preferences, under which assignment probabilities alone do not determine ex ante utility.

The broader literature examines stable random matchings under strict preferences \citep{roth1993stable} and fractional stable matchings with additive valuations \citep{caragiannis2019stable}.
Deferred acceptance with random tie-breaking is analyzed by \citet{abdulkadirouglu2009strategy}, while \citet{erdil2008whats} study its potential efficiency losses and propose efficiency improvements that preserve ex post stability.
Recent work studies ex ante fairness \citep{han2024theory}, ex ante stable lotteries with equal treatment of equals in responsive school choice \citep{okumura2026simple}, and ex ante Pareto improvements preserving ex post stability \citep{aziz2026smart}.
Incentive issues in random matching mechanisms are studied by, e.g., \citet{erdil2014strategy} and \citet{bando2025impossibility}.

On the methodological side, our analysis connects the Alkan--Gale framework with discrete convex analysis.
\citet{AlkanGale2003} establish existence and lattice results taking continuous choice functions as primitives.
In contrast, we derive continuous choice functions from discrete bundle preferences.
\citet{karzanov2024mixed} studies capacitated stable assignments with weakly ordered contracts through mixed-type choice functions in this framework.
These choice functions coincide with those induced by our tie-breaking rule under additive valuations with quotas, and thus, by \Cref{thm:ag-ds-equivalence}, his stable assignments are exactly the doubly-strong ex ante stable allocations in that case. 
In the discrete setting, \citet{FujishigeTamura2006,FujishigeTamura2007} develop two-sided discrete-concave market models, \citet{EguchiA2003Generalized} give a generalized Gale--Shapley algorithm, and \citet{MurotaYokoi2015} study the lattice of stable allocations.
Relative to this deterministic literature, we add ex ante fairness and ordinal identification of fractional choice.

There is also a large literature on one-sided random assignment.
\citet{hylland1979efficient} introduce the pseudo-market framework with cardinal utilities.
This approach has been extended to priorities \citep{he2018pseudo}, constraints \citep{echenique2021constrained}, and multi-unit demand \citep{budish2013designing}.
With ordinal preferences, \citet{bogomolnaia2001new} introduce probabilistic serial, with extensions to weak preferences \citep{katta2006solution,yilmaz2009random} and constrained settings \citep{budish2013designing,FujishigeSanoZhan2018,aziz2022vigilant,KawaseSumitaYokoi}.
Unlike this literature, we study two-sided preferences and stability.

\section{Model}\label{sec:model}
We first define discrete-concave valuations and stability for deterministic allocations.
We then extend valuations to fractional bundles and introduce ex ante stability and fairness.

\subsection{Discrete-Concave Market}\label{subsec:market}
A market is a tuple $(I,J,(v_k)_{k\in I\cup J})$, where $I=\{i_1,\ldots,i_n\}$ is the set of workers, $J=\{j_1,\ldots,j_m\}$ is the set of firms, and $A=I\times J$ is the set of worker--firm pairs.
Each pair may be assigned multiple indivisible matching units, such as shifts or positions.
For $i\in I$ and $j\in J$, let $A_i=\{i\}\times J$ and $A_j=I\times\{j\}$.
An agent's \emph{deterministic bundle} is a nonnegative integer vector indexed by her incident pairs, which we identify with her potential partners when convenient.

We write $\ZZ$ and $\RR$ for the integers and real numbers, with subscripts $+$ and $++$ indicating nonnegative and positive elements, respectively, and let $\RRbar=\RR\cup\{-\infty\}$.
Each agent $k\in I\cup J$ has a valuation $v_k\colon\ZZ^{A_k}\to\RRbar$.
The value $-\infty$ represents infeasibility, and thus the \emph{effective domain}
$\dom(v_k)\coloneqq\{x\in\ZZ^{A_k}\mid v_k(x)>-\infty\}$ is the set of feasible bundles.
Fix a positive integer $q\in\ZZ_{++}$ such that $\dom(v_k)\subseteq\{0,1,\ldots,q\}^{A_k}$ for every $k$, and assume $v_k(\bm{0})=0$.
On feasible bundles, $v_k$ is a von Neumann--Morgenstern utility function, and lotteries are evaluated by expected utility.

For any finite coordinate set $E$ and $a\in E$, let $\chi_a$ denote the unit vector in coordinate $a$, and set $\chi_{\varnothing}=\bm{0}$.
Vector inequalities are componentwise; $\bm{0}$ and $\bm{1}$ denote the all-zero and all-one vectors of the relevant dimension.
For a vector $z\in\RR^E$, let $\supp^+(z)=\{a\in E\mid z_a>0\}$ and $\supp^-(z)=\{a\in E\mid z_a<0\}$.
\begin{definition}[M$^\natural$-concavity]\label{def:Mnatural-concavity}
A valuation $v_k$ is \emph{M$^\natural$-concave} if, for all $x,y\in\dom(v_k)$ and $a\in\supp^+(x-y)$, there exists $b\in\supp^-(x-y)\cup\{\varnothing\}$ such that
\begin{align*}
v_k(x)+v_k(y)\le v_k(x-\chi_a+\chi_b)+v_k(y+\chi_a-\chi_b).
\end{align*}
\end{definition}
This exchange condition allows a unit to be transferred, possibly in exchange for another, without reducing the sum of the two bundles' values. It imposes substitutability without requiring additivity.
When $q=1$, we identify binary bundles with sets of partners.
For example, an additive valuation $v_k(X)=\sum_{a\in X}w_a$ on bundles satisfying $|X|\le c_k$, extended by $-\infty$ elsewhere, is M$^\natural$-concave.
We assume throughout that every $v_k$ is M$^\natural$-concave; further examples and structural properties appear in Appendix~\ref{subsec:valuations}.

\subsection{Deterministic Allocations}\label{subsec:integral-allocations}
A \emph{deterministic (integral) allocation} is a vector $x\in\ZZ_+^A$, where $x_{ij}$ is the number of matching units assigned to pair $(i,j)$.
For each agent $k$, write $x_k=x|_{A_k}$ for her bundle; we use the same restriction notation for fractional allocations.
The allocation is \emph{feasible} if $x_k\in\dom(v_k)$ for every $k\in I\cup J$.
Let $\cM\subseteq\{0,1,\ldots,q\}^A$ denote the finite set of feasible deterministic allocations.

Given an integer availability vector $z\in\ZZ_+^{A_k}$, agent $k$'s \emph{deterministic choice correspondence} is
\begin{align*}
\Chcd_k(z)\coloneqq\argmax\{v_k(y)\mid y\in\ZZ_+^{A_k},\ y\le z\}.
\end{align*}
All maximizers are included when the agent is indifferent.
For vectors $x,y$ of the same dimension, write $x\vee y$ and $x\wedge y$ for their componentwise maximum and minimum.

\begin{definition}\label{def:deterministic-stable}
A feasible integral allocation $x\in\ZZ_+^A$ is \emph{(ex post) stable} if
\begin{itemize}
    \item for every $k\in I\cup J$, it holds that $x_k\in \Chcd_k(x_k)$, and
    \item for every pair $(i,j)\in I\times J$, either $x_i\in\Chcd_i(x_i\vee q\chi_j)$ or $x_j\in\Chcd_j(x_j\vee q\chi_i)$.
\end{itemize}
\end{definition}
The first condition is \emph{individual rationality}: no agent gains by rejecting current matching units. The second rules out \emph{blocking pairs}, whose members both gain from expanding their mutual availability. This is the discrete-concave stability notion of \citet{FujishigeTamura2006} and \citet{MurotaYokoi2015}.

\subsection{Fractional Allocations}\label{subsec:fractional-allocations}
A \emph{fractional allocation} $\pi\in\RR_+^A$ records expected matching quantities. Since nonadditive valuations can yield different expected utilities for lotteries with the same mean, we evaluate each fractional bundle by its highest attainable expected utility.

For a finite set $B$, let $\Delta(B)$ denote the set of probability distributions on $B$.
For each $z\in\RR^{A_k}$, define the \emph{concave closure} of $v_k$ as 
\begin{align}
\overline v_k(z)\coloneqq
\max\left\{\mathbb E_{\mu}[v_k(X)]\ \middle|\
\mu\in\Delta(\dom(v_k)),\ \mathbb E_{\mu}[X]=z\right\},
\label{eq:concave-closure}
\end{align}
where $X$ is a deterministic bundle drawn according to $\mu$.
We set $\overline v_k(z)=-\infty$ if no such lottery exists.
Writing $\conv$ for the convex hull, we have $\dom(\overline v_k)=\conv(\dom(v_k))$; M$^\natural$-concavity also gives $\overline v_k(x)=v_k(x)$ for $x\in\dom(v_k)$ (\Cref{thm:local-global-closure}).
A fractional allocation $\pi$ is \emph{feasible} if $\pi_k\in\dom(\overline v_k)$ for every $k\in I\cup J$.
By linear-programming duality \citep{murota2003}, for every feasible fractional bundle $z$,
\begin{align}
\overline v_k(z)
=\min_{(p,\alpha)\in\RR^{A_k}\times\RR}
\left\{p\cdot z+\alpha\ \middle|\ p\cdot x+\alpha\ge v_k(x)\ \ (\forall x\in\dom(v_k))\right\}.
\label{eq:concave-closure-dual-agent}
\end{align}

Given a fractional availability vector $z\in\RR_+^{A_k}$, define the \emph{fractional choice correspondence} by
\begin{align*}
\Chcf_k(z)\coloneqq\argmax\{\overline v_k(y)\mid y\in\RR_+^{A_k},\ y\le z\}.
\end{align*}

A \emph{lottery allocation} is a distribution $\lambda\in\Delta(\cM)$ over feasible deterministic allocations.
It \emph{induces}, or \emph{implements}, $\pi$ if $\sum_{x\in\cM}\lambda_x x=\pi$, in which case we also call it a \emph{decomposition} of $\pi$.
Its support is $\supp(\lambda)=\{x\in\cM\mid\lambda_x>0\}$, and it is \emph{ex post stable} if every allocation in its support is stable.

\Needspace{9\baselineskip}
\subsection{Ex Ante Stability and Fairness}\label{subsec:ex-ante-fairness}\label{subsec:fairness-refinements}
Using the fractional choice correspondences, we define ex ante stability in direct analogy with stability for deterministic allocations.
\begin{definition}\label{def:ex-ante-stable}
A feasible fractional allocation $\pi\in\RR_+^A$ is called \emph{ex ante stable} if
\begin{itemize}
    \item for every $k\in I\cup J$, it holds that $\pi_k\in\Chcf_k(\pi_k)$, and
    \item for every pair $(i,j)\in I\times J$, either $\pi_i\in\Chcf_i(\pi_i\vee q\chi_j)$ or $\pi_j\in\Chcf_j(\pi_j\vee q\chi_i)$.
\end{itemize}
\end{definition}

The following example illustrates that ex post stability need not imply ex ante stability.

\begin{example}\label{ex:expost-not-exante}
Let $I=\{i_1,i_2,i_3\}$ and $J=\{j_1,j_2,j_3\}$.
Each agent has unit-demand preferences:
\begin{align*}
i_1 &: j_1 \succ_{i_1} j_2 \succ_{i_1} j_3 \succ_{i_1} \varnothing, &
i_2 &: j_2 \succ_{i_2} j_3 \succ_{i_2} j_1 \succ_{i_2} \varnothing, &
i_3 &: j_3 \succ_{i_3} j_1 \succ_{i_3} j_2 \succ_{i_3} \varnothing, \\
j_1 &: i_2 \succ_{j_1} i_3 \succ_{j_1} i_1 \succ_{j_1} \varnothing, &
j_2 &: i_3 \succ_{j_2} i_1 \succ_{j_2} i_2 \succ_{j_2} \varnothing, &
j_3 &: i_1 \succ_{j_3} i_2 \succ_{j_3} i_3 \succ_{j_3} \varnothing.
\end{align*}
Represent each ranking by utilities $3,2,1,0$ for the first, second, third, and outside options, respectively.
Let $x$ and $y$ be the integral allocations with incidence matrices and define $\pi=\tfrac{1}{2}x+\tfrac{1}{2}y$:
\[
x=
\begin{pNiceMatrix}[first-row,first-col]
     & j_1 & j_2 & j_3 \\
i_1  & 1 & 0 & 0 \\
i_2  & 0 & 1 & 0 \\
i_3  & 0 & 0 & 1
\end{pNiceMatrix},\quad
y=
\begin{pNiceMatrix}[first-row,first-col]
     & j_1 & j_2 & j_3 \\
i_1  & 0 & 0 & 1 \\
i_2  & 1 & 0 & 0 \\
i_3  & 0 & 1 & 0
\end{pNiceMatrix},\quad
\pi=
\begin{pNiceMatrix}[first-row,first-col]
     & j_1 & j_2 & j_3 \\
i_1  & 1/2 & 0 & 1/2 \\
i_2  & 1/2 & 1/2 & 0 \\
i_3  & 0 & 1/2 & 1/2
\end{pNiceMatrix}.
\]
Every worker receives her first choice in $x$, and every firm receives its first choice in $y$, and thus both allocations are stable. Their equal mixture is nevertheless blocked by $(i_1,j_2)$ and is not ex ante stable.
\end{example}

In addition to ex ante stability, we consider three ex ante fairness requirements. Ex ante stability rules out a local reallocation toward a pair when both sides strictly benefit, while these fairness requirements address the three remaining cases in which neither side loses and one or both sides are indifferent.

If the firm is indifferent while the worker strictly benefits, underassignment to the pair may reflect discrimination among workers.
\begin{definition}
A feasible fractional allocation $\pi\in\RR_+^A$ has no \emph{ex ante discrimination among workers} if there are no $i\in I$, $i'\in I\cup\{\varnothing\}$, $j\in J$, $j'\in J\cup\{\varnothing\}$, and $\epsilon>0$ such that $\overline{v}_i(\pi_i+\epsilon(\chi_j-\chi_{j'}))> \overline{v}_i(\pi_i)$, $\overline{v}_j(\pi_j+\epsilon(\chi_i-\chi_{i'}))=\overline{v}_j(\pi_j)$, and $\pi_{ij}<\pi_{i'j}$ if $i'\in I$.
\end{definition}

The symmetric case, in which the worker is indifferent while the firm strictly benefits, may reflect discrimination among firms.
\begin{definition}
A feasible fractional allocation $\pi\in\RR_+^A$ has no \emph{ex ante discrimination among firms} if there are no $i\in I$, $i'\in I\cup\{\varnothing\}$, $j\in J$, $j'\in J\cup\{\varnothing\}$, and $\epsilon>0$ such that $\overline{v}_i(\pi_i+\epsilon(\chi_j-\chi_{j'}))= \overline{v}_i(\pi_i)$, $\overline{v}_j(\pi_j+\epsilon(\chi_i-\chi_{i'}))>\overline{v}_j(\pi_j)$, and $\pi_{ij}<\pi_{ij'}$ if $j'\in J$.
\end{definition}

If both sides are indifferent between the relevant partners, the pair should not be assigned less than the pairs from which the reallocation would draw; this motivates indifference neutrality.
\begin{definition}
A feasible fractional allocation $\pi\in\RR_+^A$ is \emph{ex ante indifference neutral} if there are no $i\in I$, $i'\in I\cup\{\varnothing\}$, $j\in J$, $j'\in J\cup\{\varnothing\}$, and $\epsilon>0$ such that $\overline{v}_i(\pi_i+\epsilon(\chi_j-\chi_{j'}))= \overline{v}_i(\pi_i)$, $\overline{v}_j(\pi_j+\epsilon(\chi_i-\chi_{i'}))= \overline{v}_j(\pi_j)$, $\pi_{ij}<\pi_{i'j}$ if $i'\in I$, and $\pi_{ij}<\pi_{ij'}$ if $j'\in J$.
\end{definition}

These conditions yield our main solution concept.
\begin{definition}\label{def:ds-ex-ante-fair}
A feasible fractional allocation $\pi\in\RR_+^A$ is \emph{doubly-strong ex ante stable} if it satisfies ex ante stability, no ex ante discrimination among firms, no ex ante discrimination among workers, and ex ante indifference neutrality.
\end{definition}

The following example illustrates how doubly-strong ex ante stability combines stability with fairness.
\begin{example}\label{ex:local-equal-treatment}
Let $I=\{i_1\}$ and $J=\{j_1,j_2\}$, with unit demand for every agent.
The worker is indifferent between the two firms and strictly prefers either match to remaining unmatched, while each firm strictly prefers matching with the worker.
No ex ante discrimination among firms requires $\pi_{i_1j_1}\geq\pi_{i_1j_2}$ and, symmetrically, $\pi_{i_1j_2}\geq\pi_{i_1j_1}$.
Ex ante stability requires $\pi_{i_1j_1}+\pi_{i_1j_2}=1$.
Hence, the unique doubly-strong ex ante stable fractional allocation assigns probability $1/2$ to each match.
\end{example}

Our definition coincides with the stability and fairness notions studied by \citet{kesten2015theory} and \citet{cookson2025fairly} once their ordinal preferences are represented by additive valuations.

\begin{proposition}
\label{prop:responsive-special-cases}
Represent each agent's weak order over partners by an additive valuation whose weights are positive and consistent with that order, subject to unit demand for students, proposers, and acceptors, and to the capacity constraint for each school.
\begin{enumerate}[(i)]
\item In the model of \citet{kesten2015theory}, with strict student preferences, weak school priorities, total capacity equal to the number of students, and all schools acceptable to all students, a fractional allocation is doubly-strong ex ante stable if and only if it is a random matching in their sense that is strongly ex ante stable.
\item In the model of \citet{cookson2025fairly}, with weak preferences on both sides, equally many agents on the two sides, and all agents mutually acceptable, a fractional allocation is doubly-strong ex ante stable if and only if it is a fractional perfect matching that is doubly-strong ex ante stable in their sense.
\end{enumerate}
\end{proposition}
\begin{proof}
Under these representations, each concave closure is linear on the relevant unit-demand or capacity polytope, and thus a small transfer $\epsilon(\chi_a-\chi_{a'})$ away from a partner $a'$ with positive weight is utility improving, utility preserving, or utility worsening according to whether $a$ is ranked above, equally to, or below $a'$.
Every doubly-strong ex ante stable allocation fully assigns every agent; otherwise, since the two sides' capacities balance, some underassigned agents on the two sides would both strictly benefit from increasing their mutual assignment.
At a fully assigning allocation, every deviation in our definitions must displace partners $i'\in I$ and $j'\in J$ with positive weight.
Our four conditions then reduce to the absence of ex ante justified envy and ex ante discrimination in (i), and to Definitions~1, 2, 4, and~5 of \citet{cookson2025fairly} in (ii); in (i), strict student preferences make the last two conditions vacuous.
\end{proof}

\section{Existence and Stable Lottery Implementation}\label{sec:cardinal}
This section establishes existence, lattice structure, and stable lottery implementation for discrete-concave markets.
We first derive choice functions from our choice correspondences and connect our stability notion to Alkan--Gale stability to obtain existence and lattice results.
We then show that ex ante stable fractional allocations can be implemented by lotteries over stable deterministic allocations while preserving every agent's assigned utility.

\subsection{Alkan--Gale Framework}\label{subsec:alkan-gale}\label{subsec:choice-correspondences-functions}
We begin with the choice-function framework of \citet{AlkanGale2003}, which provides existence and lattice results.
Unlike a choice correspondence, a \emph{choice function} selects a single fractional bundle: on a finite coordinate set $E$, it is a map $C\colon\RR_+^E\to\RR_+^E$ satisfying $C(x)\le x$ for every $x$.

We first state the conditions required by this framework.
A choice function $C$ is \emph{continuous} if $x^m\to x$ implies $C(x^m)\to C(x)$.
It satisfies \emph{consistency}, \emph{persistence}, and \emph{size monotonicity}, respectively, if
\begin{align*}
C(x)\le y\le x &\ \Longrightarrow\ C(y)=C(x), &&\text{(consistency)}\\
y\le x &\ \Longrightarrow\ C(y)\ge C(x)\wedge y, &&\text{(persistence)}\\
y\le x &\ \Longrightarrow\ \|C(y)\|_1\le\|C(x)\|_1. &&\text{(size monotonicity)}
\end{align*}
Consistency says that removing unchosen availability does not affect choice.\footnote{This condition is also known as \emph{irrelevance of rejected contracts}; see \citet{aygun2013matching}.}
Persistence is the fractional analogue of substitutability, while size monotonicity is also known as the law of aggregate demand.

In this subsection, consider a market $(I,J,(C_k)_{k\in I\cup J})$, where each $C_k$ is a choice function on $\RR_+^{A_k}$.
We assume that $C_k(x)_a\le q$ for every $k$, $x$, and $a\in A_k$.
The functions constructed below depend only on $x\wedge q\bm{1}$, and thus this bound is consistent with allowing arbitrary nonnegative availability vectors.

\begin{definition}[AG-preference]
For an agent $k\in I\cup J$ with choice function $C_k\colon\RR_+^{A_k}\to\RR_+^{A_k}$, and for any two vectors $x,x'\in\RR_+^{A_k}$, we write $x\succeq_k^{\rm AG} x'$ if $C_k(x\vee x')=x$.
\end{definition}

Let $\pi$ be a fractional allocation.
For $i\in I$ and $j\in J$, we say that $\pi_i$ is \emph{$j$-satiated} if $C_i(x)_j\le \pi_{ij}$ for all $x\ge \pi_i$.
Similarly, we say that $\pi_j$ is \emph{$i$-satiated} if $C_j(x)_i\le \pi_{ij}$ for all $x\ge \pi_j$.

\Needspace{7\baselineskip}
\begin{definition}
A fractional allocation $\pi\in\RR_+^{A}$ is \emph{AG-stable} if the following two conditions hold:
\begin{itemize}
    \item for every $k\in I\cup J$, it holds that $\pi_k=C_k(\pi_k)$, and
    \item for every pair $(i,j)\in I\times J$, either $\pi_i$ is $j$-satiated or $\pi_j$ is $i$-satiated.
\end{itemize}
\end{definition}

The following lemma shows that non-satiation in coordinate $a$ can be detected by increasing only the availability of coordinate $a$ from the current choice.
\begin{lemma}\label{lem:nonsatiation-expansion}
Let $C\colon\RR^{A}_{+}\to \RR^{A}_{+}$ be a consistent and persistent choice function such that $C(x)\le q\bm{1}$ for every $x\in\RR^{A}_{+}$.
Fix $\zeta\in\RR^{A}_{+}$ and $a\in A$, and suppose that $C(\zeta)=\zeta$.
If $\zeta$ is not $a$-satiated, then $C(\zeta\vee q\chi_a)_a>\zeta_a$.
\end{lemma}
\begin{proof}
Since $\zeta$ is not $a$-satiated, there exists $x\ge \zeta$ such that $C(x)_a>\zeta_a$.
Define $y\coloneqq\zeta\vee C(x)_a\chi_a$.
Then, $x\ge y\ge \zeta$ and $y\le \zeta\vee q\chi_a$, because $C(x)_a\le q$.
By persistence applied to $x\ge y$, we have $C(y)\ge C(x)\wedge y$.
Since $y_a=C(x)_a$, this implies
\begin{align}
  C(y)_a\ge C(x)_a>\zeta_a. \label{eq:nonsatiation-expansion}
\end{align}

Now let $\xi\coloneqq\zeta\vee q\chi_a$.
Suppose, toward a contradiction, that $C(\xi)_a\le \zeta_a$.
Because $\xi_b=\zeta_b$ for every $b\neq a$ and $C(\xi)\le \xi$, we have $C(\xi)\le \zeta$.
Together with $C(\zeta)=\zeta\le \xi$, this gives $C(\xi)=\zeta$ by consistency.
Since $C(\xi)=\zeta \le y \le \xi$, consistency implies $C(y)=C(\xi)=\zeta$, which contradicts \eqref{eq:nonsatiation-expansion}.
Therefore, $C(\zeta\vee q\chi_a)_a>\zeta_a$.
\end{proof}

The next theorem establishes the existence of AG-stable allocations under these conditions.
\begin{theorem}[\cite{AlkanGale2003}]\label{thm:alkan-gale}
For any allocation problem where every agent $k\in I\cup J$ has a continuous, persistent, and consistent choice function $C_k$, there exists an AG-stable allocation $\pi^*$ that dominates all other AG-stable allocations in terms of the AG-preferences of the agents in $I$.
Moreover, if every choice function $C_k$ additionally satisfies size monotonicity, then the set of all AG-stable allocations forms a distributive lattice under the orderings $(\succeq_i^{\rm AG})_{i\in I}$ and $(\succeq_j^{\rm AG})_{j\in J}$.
\end{theorem}

\citet{AlkanGale2003} established this result via a deferred-acceptance iteration, summarized in \Cref{alg:AG-DA}.

\begin{algorithm}[ht]
\caption{Deferred-Acceptance Iteration of \citet{AlkanGale2003}}\label{alg:AG-DA}
Let $\pi^{(0)}$ be the fractional allocation such that $\pi^{(0)}_{ij}=q$ for all $i\in I$ and $j\in J$\;
\For{$t\gets 1,2,\dots$}{
    Let $\mu^{(t)}$ be the fractional allocation such that $\mu^{(t)}_i=C_i(\pi_i^{(t-1)})$ for each $i\in I$\;
    Let $\nu^{(t)}$ be the fractional allocation such that $\nu^{(t)}_j=C_j(\mu_j^{(t)})$ for each $j\in J$\;
    \lIf{$\mu^{(t)}=\nu^{(t)}$}{\KwRet $\mu^{(t)}$}
    Let $\pi^{(t)}$ be the fractional allocation such that $\pi^{(t)}_{ij}=\pi^{(t-1)}_{ij}$ if $\mu^{(t)}_{ij}=\nu^{(t)}_{ij}$ and $\pi^{(t)}_{ij}=\nu^{(t)}_{ij}$ if $\nu^{(t)}_{ij}<\mu^{(t)}_{ij}$ for each $(i,j)\in I\times J$\;
}
\end{algorithm}
If the iteration returns at some finite step $t$, its output is the common allocation $\mu^{(t)}=\nu^{(t)}$.
If it never returns, its outcome is defined by $\pi^*\coloneqq
\lim_{t\to\infty}\mu^{(t)}=\lim_{t\to\infty}\nu^{(t)}$.
In either case, the resulting allocation is the $I$-optimal AG-stable allocation of \Cref{thm:alkan-gale}.
\subsection{Fair Tie-Breaking and Induced Choice Functions}\label{subsec:derive-choice-functions}

The Alkan--Gale framework takes single-valued choice functions as primitives, whereas our valuations generate fractional choice correspondences. To apply its existence and lattice results, we must therefore resolve indifferences within our fractional choice correspondences. This tie-breaking must serve two purposes. First, the induced choice functions must satisfy the Alkan--Gale conditions. Second, stability under these functions must capture the stability and fairness requirements in our original model. We construct such tie-breaking rules based on \emph{Lorenz dominance}.

Consider an M$^\natural$-concave valuation $g\colon\ZZ^A\to\RRbar$ with $\bm{0}\in\dom(g)\subseteq\{0,1,\ldots,q\}^A$, and let $\overline g$ be its concave closure. Given availability $x\in\RR_+^A$, write
\[
\Chcf(x)\coloneqq\argmax\{\overline g(y)\mid y\le x,\ y\in\dom(\overline g)\}.
\]
\paragraph{Lorenz dominance.}
For $y\in\RR_+^A$, let $y^\uparrow$ denote its coordinates sorted in nondecreasing order. For any $y,z\in\RR_+^A$, we say that $y$ \emph{Lorenz-dominates} $z$ if
\[
\sum_{\ell=1}^{r}y^\uparrow_\ell
\ge
\sum_{\ell=1}^{r}z^\uparrow_\ell
\qquad(r=1,\ldots,|A|).
\]
A fractional bundle in $\Chcf(x)$ is \emph{Lorenz-dominant} if it Lorenz-dominates every fractional bundle in $\Chcf(x)$.
Taking $r=|A|$ shows that such a fractional bundle automatically maximizes total assignment.\footnote{Note that Lorenz dominance allows us to compare fractional bundles with different total quantities.}
Both an increase in one coordinate and a small transfer from a more-assigned partner to a less-assigned partner, without reversing their relative assignment levels, improve the fractional bundle in the Lorenz ranking.
Within $\Chcf(x)$, these changes preserve utility and correspond to the improvements considered by our stability and fairness requirements.\footnote{
Lorenz comparisons are central to the analysis of income inequality; see \citet{Atkinson1970inequality}.
In cooperative game theory, \citet{DuttaRay1989egalitarian} use them to formulate egalitarian allocations subject to participation constraints.
Related distributional comparisons also appear in school choice, where \citet{echenique2024diversity} study diversity and representation across student groups.}

For example, a unit-demand agent who regards three partners as equally desirable and strictly prefers each to remaining unmatched selects $(1/3,1/3,1/3)$ when all three are fully available. With availabilities $(1/5,1,1)$, the rule instead selects $(1/5,2/5,2/5)$.

\paragraph{Induced Choice Functions.}
Lorenz dominance is generally incomplete, and thus a fractional bundle that dominates every alternative need not exist.
The key consequence of M$^\natural$-concavity is that our fractional choice correspondence has additional structure.
By \Cref{thm:local-global-closure,thm:Mconcave-g}, $\Chcf(x)$ is a nonempty compact generalized polymatroid.
We write $\Chf(x)$ for the Lorenz-dominant fractional bundle in $\Chcf(x)$, whose existence and uniqueness are established below.
Define the tie-breaking objective $d(y)\coloneqq\sum_{a\in A}(q-y_a)^2$.
Using the results of \citet{Fujishige1980lex} on lexicographically optimal bases, we obtain the following proposition, proved in Appendix~\ref{app:section3-proofs}.

\begin{proposition}\label{prop:lorenz-choice}
For every $x\in\RR_+^A$, $\Chcf(x)$ contains a unique Lorenz-dominant fractional bundle $\Chf(x)$.
This fractional bundle uniquely maximizes $y^\uparrow$ in lexicographic order over $y\in\Chcf(x)$.
It also uniquely minimizes $d(y)$ over $y\in\Chcf(x)$.
\end{proposition}

The induced choice function can therefore be written as
\[
\{\Chf(x)\}=\argmin\{d(y)\mid y\in\Chcf(x)\}.
\]

For each agent $k$, we apply this construction to $g=v_k$ and write $\Chf_k$. We now verify the two requirements.
\begin{lemma}\label{lem:induced-choice-regularity}
Let $g\colon\ZZ^A\to\RRbar$ be M$^\natural$-concave with $\bm{0}\in\dom(g)\subseteq\{0,1,\ldots,q\}^A$.
Then, the induced choice function $\Chf\colon\RR_+^A\to\RR_+^A$ is continuous, consistent, persistent, and size-monotone.
\end{lemma}

\begin{lemma}\label{lem:non-satiated}
Let $g\colon\ZZ^A\to\RRbar$ be M$^\natural$-concave with $\bm{0}\in\dom(g)\subseteq\{0,1,\ldots,q\}^A$, and let $\Chf\colon\RR_+^A\to\RR_+^A$ be the choice function induced by $g$.
For $\zeta\in\RR_+^A$ and $a\in A$, let $\zeta'=\Chf(\zeta\vee q\chi_a)$, and suppose that $\Chf(\zeta)=\zeta$ and $\zeta\ne\zeta'$.
Then, $\zeta'_a>\zeta_a$.
Moreover, there exist $a'\in A\cup\{\varnothing\}$ and $\epsilon_0>0$ such that (i) $\overline{g}(\zeta+\epsilon(\chi_a-\chi_{a'}))>\overline{g}(\zeta)~(\forall\epsilon\in(0,\epsilon_0])$, or (ii) $\overline{g}(\zeta+\epsilon(\chi_a-\chi_{a'}))=\overline{g}(\zeta)~(\forall \epsilon\in(0,\epsilon_0])$ and $\zeta_{a'}>\zeta_a$ if $a'\in A$.
\end{lemma}

The proofs of \Cref{lem:induced-choice-regularity,lem:non-satiated} are given in Appendix~\ref{app:section3-proofs}.
The first lemma ensures that the induced choice functions fall within the Alkan--Gale framework.\footnote{\citet{MurotaYokoi2015} establish the corresponding discrete choice properties for a unique-selecting M$^\natural$-concave valuation. Here uniqueness is obtained by resolving ties among fractional utility maximizers.}
The second identifies the utility and assignment comparisons underlying their non-satiation conditions.
In the next subsection, we use these comparisons to establish the equivalence between AG-stability and doubly-strong ex ante stability.

\subsection{Equivalence with Doubly-Strong Ex Ante Stability}\label{subsec:ds-ex-ante-equivalence}

Given a market $(I,J,(v_k)_{k\in I\cup J})$, we define the AG-market $(I,J,(\Chf_k)_{k\in I\cup J})$ using the induced choice functions.
We show that AG-stability and doubly-strong ex ante stability coincide.

\begin{theorem}\label{thm:ag-ds-equivalence}
A fractional allocation $\pi$ is AG-stable in $(I,J,(\Chf_k)_{k\in I\cup J})$ if and only if it is doubly-strong ex ante stable in $(I,J,(v_k)_{k\in I\cup J})$.
\end{theorem}
\begin{proof}
We first prove the forward direction.
Assume that $\pi$ is not doubly-strong ex ante stable.
We show that $\pi$ is not AG-stable.

First, suppose that $\pi$ is not ex ante stable.
If there is an agent $k\in I\cup J$ such that $\pi_k\notin\argmax_{\xi\le\pi_k}\overline{v}_k(\xi)$, then $\pi_k\ne\Chf_k(\pi_k)$ and $\pi$ is not AG-stable.
Otherwise, every agent maximizes under the current cap, and there is a pair $(i,j)\in I\times J$ for which neither side maximizes under the corresponding expanded cap.
The expanded cap for worker $i$ differs from $\pi_i$ only in coordinate $j$.
Since its optimal value is strictly larger than the old-cap optimal value, every expanded-cap maximizer has its $j$th coordinate strictly larger than $\pi_{ij}$.
Thus, $\Chf_i(\pi_i\vee q\chi_j)_j>\pi_{ij}$, and symmetrically $\Chf_j(\pi_j\vee q\chi_i)_i>\pi_{ij}$.
The pair violates AG-stability.

We may therefore assume that $\pi$ is ex ante stable.
In particular, $\pi_k\in\Chcf_k(\pi_k)$ for every agent $k$.
Since every $y\in\Chcf_k(\pi_k)$ lies in $[\bm{0},\pi_k]\subseteq[0,q]^{A_k}$, it satisfies $d(y)\ge d(\pi_k)$, with strict inequality unless $y=\pi_k$.
It follows that
\[
\Chf_k(\pi_k)=\pi_k
\qquad(\forall k\in I\cup J).
\]

We use the following observation.
Fix an agent $k$, a coordinate $a$, and a point $y\in\dom(\overline v_k)$ with $\bm{0}\le y\le\pi_k\vee q\chi_a$ such that
\[
\overline{v}_k(y)\ge\overline{v}_k(\pi_k)
\quad\text{and}\quad
d(y)<d(\pi_k).
\]
Then, $\Chf_k(\pi_k\vee q\chi_a)_a>(\pi_k)_a$.
Indeed, if the selected expanded-cap maximizer had its $a$th coordinate at most $(\pi_k)_a$, it would lie below $\pi_k$.
Old-cap optimality would then force the expanded optimal value to equal $\overline{v}_k(\pi_k)$, making both that selected point and $y$ expanded-cap maximizers.
The expanded-cap tie-break would give $d(\Chf_k(\pi_k\vee q\chi_a))\le d(y)<d(\pi_k)$, whereas $\Chf_k(\pi_k)=\pi_k$ gives the reverse weak inequality for any old-cap maximizer, a contradiction.
We will also use the simpler fact that if $\overline{v}_k(y)>\overline{v}_k(\pi_k)$, then the same coordinate increase follows directly, because any point whose $a$th coordinate is at most $(\pi_k)_a$ lies under the old cap.

Suppose that $\pi$ has ex ante discrimination among firms, witnessed by $i\in I$, $i'\in I\cup\{\varnothing\}$, $j\in J$, $j'\in J\cup\{\varnothing\}$, and $\epsilon>0$.
The strict improvement of firm $j$ implies
\[
\Chf_j(\pi_j\vee q\chi_i)_i>\pi_{ij}.
\]
For worker $i$, set
\[
\delta_i\coloneqq
\begin{cases}
\epsilon & \text{if }j'=\varnothing,\\
\min\{\epsilon,(\pi_{ij'}-\pi_{ij})/2\}
& \text{if }j'\in J.
\end{cases}
\]
Then, $\delta_i>0$, and concavity gives
\[
\overline{v}_i(\pi_i+\delta_i(\chi_j-\chi_{j'}))
\ge\overline{v}_i(\pi_i).
\]
Moreover, $d(\pi_i+\delta_i(\chi_j-\chi_{j'}))<d(\pi_i)$: if $j'\in J$, this follows from $\delta_i<\pi_{ij'}-\pi_{ij}$, while if $j'=\varnothing$, it follows because the deviation strictly increases coordinate $j$ within $[0,q]$.
The observation therefore gives $\Chf_i(\pi_i\vee q\chi_j)_j>\pi_{ij}$.
Thus, neither side is satiated, and $\pi$ is not AG-stable.
The case of ex ante discrimination among workers is symmetric, including the case $i'=\varnothing$.

Finally, suppose that $\pi$ is not ex ante indifference neutral.
Let the violation be witnessed by $i\in I$, $i'\in I\cup\{\varnothing\}$, $j\in J$, $j'\in J\cup\{\varnothing\}$, and $\epsilon>0$.
Define
\[
\delta_i\coloneqq
\begin{cases}
\epsilon & \text{if }j'=\varnothing,\\
\min\{\epsilon,(\pi_{ij'}-\pi_{ij})/2\}
& \text{if }j'\in J,
\end{cases}
\qquad
\delta_j\coloneqq
\begin{cases}
\epsilon & \text{if }i'=\varnothing,\\
\min\{\epsilon,(\pi_{i'j}-\pi_{ij})/2\}
& \text{if }i'\in I.
\end{cases}
\]
Both numbers are positive.
On each side, concavity makes the utility at the shortened deviation weakly greater than the current utility, and the definitions of $\delta_i$ and $\delta_j$ make the corresponding $d$-values strictly smaller.
The observation applies separately to the two sides, including when one or both displaced partners are the outside option, and yields
\[
\Chf_i(\pi_i\vee q\chi_j)_j>\pi_{ij},
\qquad
\Chf_j(\pi_j\vee q\chi_i)_i>\pi_{ij}.
\]
Hence, $\pi$ is not AG-stable.
This proves the forward direction.

\medskip
We now prove the reverse direction.
Assume that $\pi$ is not AG-stable.
If there is an agent $k\in I\cup J$ such that $\pi_k\ne \Chf_k(\pi_k)$, let $\xi^*=\Chf_k(\pi_k)$.
If $\overline{v}_k(\xi^*)>\overline{v}_k(\pi_k)$, then $\pi_k\notin\argmax_{\xi\le\pi_k}\overline{v}_k(\xi)$ and $\pi$ is not ex ante stable.
Otherwise, we have $\overline{v}_k(\xi^*)=\overline{v}_k(\pi_k)$, and since $\pi_k\in\Chcf_k(\pi_k)$ and $\xi^*\ne\pi_k$, the tie-breaking rule in the definition of $\Chf$ implies $d(\xi^*)<d(\pi_k)$.
Since $\xi^*\le\pi_k\le q\bm{1}$ and $d$ is componentwise nonincreasing on $[0,q]^A$, we have $d(\xi^*)\ge d(\pi_k)$.
This is a contradiction.
Hence, we assume that $\pi_k=\Chf_k(\pi_k)$ for all $k\in I\cup J$.

Since $\pi$ is not AG-stable, there is a pair $(i,j)\in I\times J$ such that $\pi_i$ is not $j$-satiated and $\pi_j$ is not $i$-satiated.
By \Cref{lem:nonsatiation-expansion}, we have $\pi_{ij}<\Chf_i(\pi_i\vee q\chi_j)_j$ and $\pi_{ij}<\Chf_j(\pi_j\vee q\chi_i)_i$.
In particular, both expanded-cap choices differ from the corresponding current allocations.
Hence, \Cref{lem:non-satiated} applies to each side.
Let $\epsilon$ be the minimum of the two step sizes supplied by that lemma.
Then, there exist $i'\in I\cup\{\varnothing\}$ and $j'\in J\cup\{\varnothing\}$ such that one of the following four conditions holds:
\begin{itemize}
\item[(i)] $\overline{v}_i(\pi_i)<\overline{v}_i(\pi_i+\epsilon(\chi_j-\chi_{j'}))$ and $\overline{v}_j(\pi_j)<\overline{v}_j(\pi_j+\epsilon(\chi_i-\chi_{i'}))$.
\item[(ii)] $\overline{v}_i(\pi_i)<\overline{v}_i(\pi_i+\epsilon(\chi_j-\chi_{j'}))$, $\overline{v}_j(\pi_j)=\overline{v}_j(\pi_j+\epsilon(\chi_i-\chi_{i'}))$, and $\pi_{ij}<\pi_{i'j}$ if $i'\in I$.
\item[(iii)] $\overline{v}_i(\pi_i)=\overline{v}_i(\pi_i+\epsilon(\chi_j-\chi_{j'}))$, $\overline{v}_j(\pi_j)<\overline{v}_j(\pi_j+\epsilon(\chi_i-\chi_{i'}))$, and $\pi_{ij}<\pi_{ij'}$ if $j'\in J$.
\item[(iv)] $\overline{v}_i(\pi_i)=\overline{v}_i(\pi_i+\epsilon(\chi_j-\chi_{j'}))$, $\overline{v}_j(\pi_j)=\overline{v}_j(\pi_j+\epsilon(\chi_i-\chi_{i'}))$, $\pi_{ij}<\pi_{i'j}$ if $i'\in I$, and $\pi_{ij}<\pi_{ij'}$ if $j'\in J$.
\end{itemize}
If (i) holds, then $\pi$ is not ex ante stable.
If (ii) holds, then $\pi$ has ex ante discrimination among workers.
If (iii) holds, then $\pi$ has ex ante discrimination among firms.
If (iv) holds, then $\pi$ is not ex ante indifference neutral.
Therefore, $\pi$ is not doubly-strong ex ante stable.
This completes the proof.
\end{proof}

The following theorem summarizes the existence and lattice results.
\begin{theorem}
\label{thm:cardinal-main}
Suppose that every agent's valuation is M$^\natural$-concave. Then, the set of doubly-strong ex ante stable allocations is nonempty and forms a distributive lattice under the induced AG-preference orders on the two sides.
\end{theorem}
\begin{proof}
By \Cref{lem:induced-choice-regularity}, the induced choice functions are continuous, consistent, persistent, and size-monotone.
\Cref{thm:alkan-gale} therefore gives a nonempty distributive lattice of AG-stable allocations under the induced AG-preference orders.
By \Cref{thm:ag-ds-equivalence}, this is exactly the set of doubly-strong ex ante stable allocations.
\end{proof}

\subsection{From Ex Ante to Ex Post Stability}\label{subsec:ex-ante-to-ex-post}

We now turn to lottery implementation.
We first show that every feasible fractional allocation admits a common lottery that delivers to each agent the utility assigned to her fractional bundle.

\begin{proposition}
\label{prop:joint-implementability}
For every feasible fractional allocation $\pi\in\RR_+^{A}$, there exists a lottery allocation $\lambda\in\Delta(\cM)$ that induces $\pi$ and simultaneously delivers to each agent the utility assigned to her fractional bundle:
\[
\sum_{x\in\cM}\lambda_x v_k(x_k)=\overline{v}_k(\pi_k)
\qquad(\forall k\in I\cup J).
\]
\end{proposition}
The difficulty is that each agent's utility is defined by a separate maximization, whereas implementation requires a single lottery for the whole market.
The concave closure allows us to linearize each agent's problem: a supporting hyperplane identifies deterministic bundles on which the agent's nonlinear valuation agrees with an agent-specific affine function.
For lotteries supported on these bundles, expected utility is therefore determined solely by the induced fractional bundle.
To combine these bundle restrictions across the two sides of the market, we use an intersection theorem for M$^\natural$-convex sets.
We first define these sets.

\begin{definition}[M$^\natural$-convex set]\label{def:Mnatural-convex-set}
A nonempty set $Q\subseteq\ZZ^E$ is \emph{M$^\natural$-convex} if, for all $x,y\in Q$ and $a\in\supp^+(x-y)$, there exists $b\in\supp^-(x-y)\cup\{\varnothing\}$ such that both $x-\chi_a+\chi_b$ and $y+\chi_a-\chi_b$ belong to $Q$.
\end{definition}
This is the set version of the exchange condition in \Cref{def:Mnatural-concavity}.
M$^\natural$-convex sets have several useful properties (Appendix~\ref{subsec:valuations}).
First, the maximizer set of an M$^\natural$-concave function is M$^\natural$-convex.
Second, the direct product of M$^\natural$-convex sets on disjoint coordinate sets is also M$^\natural$-convex.

\begin{theorem}[{\cite{edmonds1970submodular,Frank1984c}; see also Theorem~4.22 of \cite{murota2003}}]\label{thm:Mintersection}
For two M$^\natural$-convex sets $Q_1,Q_2\subseteq \ZZ^A$, we have
\begin{align*}
    \conv(Q_1)\cap \conv(Q_2)=\conv(Q_1\cap Q_2).
\end{align*}
\end{theorem}
This generalizes the Birkhoff--von Neumann decomposition theorem.
In our application, $Q_1$ and $Q_2$ encode the utility-preserving bundles for the two sides of the market.
The equality turns separate decompositions into a common lottery satisfying both sides' bundle restrictions.

\begin{proof}[Proof of \Cref{prop:joint-implementability}]
Fix a feasible fractional allocation $\pi$.
For each agent $k$, the maximization in \eqref{eq:concave-closure} is a feasible, bounded linear program over the finite set $\dom(v_k)$, and thus it has an optimal lottery $\mu^k$ with mean $\pi_k$.
By linear-programming duality, there are $p_k\in\RR^{A_k}$ and $\alpha_k\in\RR$ such that
\begin{align*}
p_k\cdot z+\alpha_k&\ge v_k(z) &&(z\in\dom(v_k)),\\
p_k\cdot\pi_k+\alpha_k&=\overline v_k(\pi_k).
\end{align*}
These coefficients define a supporting hyperplane of the concave closure at $\pi_k$; the dual representation is given in \eqref{eq:concave-closure-dual-agent}.
Let
\[
D_k\coloneqq\{z\in\dom(v_k)\mid v_k(z)=p_k\cdot z+\alpha_k\}.
\]
Every dual slack is nonnegative, and its average under $\mu^k$ is zero:
\[
\sum_{z\in\dom(v_k)}\mu^k_z\bigl(p_k\cdot z+\alpha_k-v_k(z)\bigr)
=p_k\cdot\pi_k+\alpha_k-\overline v_k(\pi_k)=0.
\]
Thus, $\mu^k$ is supported on $D_k$, which is nonempty, and $\pi_k\in\conv(D_k)$.
Moreover,
\[
D_k=\argmax_{z\in\ZZ^{A_k}}\{v_k(z)-p_k\cdot z\}.
\]
Subtracting a linear function preserves M$^\natural$-concavity, and thus its exchange condition implies that $D_k$ is M$^\natural$-convex.

Identify the products $D_I\coloneqq\prod_{i\in I}D_i$ and $D_J\coloneqq\prod_{j\in J}D_j$ with subsets of $\ZZ^A$.
They are M$^\natural$-convex, and $\pi_k\in\conv(D_k)$ for every $k$ implies
\[
\pi\in\conv(D_I)\cap\conv(D_J)=\conv(D_I\cap D_J),
\]
where the equality is \Cref{thm:Mintersection}.
Hence, there is a lottery $\lambda$ with mean $\pi$ and support contained in $D_I\cap D_J\subseteq\cM$.
For every agent $k$, the valuation $v_k$ agrees with the affine function $z\mapsto p_k\cdot z+\alpha_k$ at all support bundles, and hence
\begin{align*}
\sum_{x\in\cM}\lambda_x v_k(x_k)
&=\sum_{x\in\cM}\lambda_x(p_k\cdot x_k+\alpha_k)
=p_k\cdot\pi_k+\alpha_k
=\overline v_k(\pi_k).
\end{align*}
The same lottery therefore implements $\pi$ and preserves every agent's utility.
\end{proof}

We next show that ex ante stability allows the implementing lottery to be supported entirely on stable deterministic allocations.
The proof is given in Appendix~\ref{app:section3-proofs}.

\begin{theorem}
\label{thm:ex-ante-to-ex-post}
Let $\pi\in\RR_+^{A}$ be an ex ante stable fractional allocation.
Then, there exists a lottery allocation $\lambda\in\Delta(\cM)$ that induces $\pi$ and simultaneously attains every agent's concave-closure utility.
Moreover, every such lottery is ex post stable.
\end{theorem}

Existence follows from \Cref{prop:joint-implementability}; the substance of the theorem is that ex ante stability, combined with utility preservation, rules out unstable realizations.
Thus, once a utility-preserving implementation is selected, ex post stability follows automatically.

The following example shows that without M$^\natural$-concavity, an ex ante stable fractional allocation can fail to be decomposable.
\begin{example}\label{ex:non-decomposable}
Let $I=\{i_1,i_2,i_3\}$ and $J=\{j_1,j_2\}$.
Set $q=1$.
For each $i\in I$, let $v_i(x)=0$ if $x\in\{0,1\}^{J}$ and $x_{j_1}+x_{j_2}\le 1$, and let $v_i(x)=-\infty$ otherwise.
For each firm $j\in J$, we identify its bundle with a vector $x\in\ZZ^I$.
Let
\begin{align*}
\mathcal{F}_{j_1}&=\{(0,0,0),(1,0,0),(0,1,0),(0,0,1),(1,1,0)\},\\
\mathcal{F}_{j_2}&=\{(0,0,0),(1,0,0),(0,1,0),(0,0,1),(0,1,1)\},
\end{align*}
and define $v_{j}(x)=0$ if $x\in\mathcal{F}_j$, and $v_j(x)=-\infty$ otherwise.

Consider the fractional allocation $\pi$ with $\pi_{ij}=\tfrac{1}{2}$ for all $(i,j)\in I\times J$.
Each worker $i\in I$ has unit demand, and $v_i$ is constant on the feasible set.
Thus, $\pi_i\in\Chcf_i(\pi_i)$ and also $\pi_i\in\Chcf_i(\pi_i\vee q\chi_j)$ for every $j\in J$.
For the firms,
\begin{align*}
\pi_{j_1}=\tfrac{1}{2}(1,1,0)+\tfrac{1}{2}(0,0,1)
\quad\text{and}\quad
\pi_{j_2}=\tfrac{1}{2}(1,0,0)+\tfrac{1}{2}(0,1,1).
\end{align*}
Hence, $\pi_j\in\Chcf_j(\pi_j)$ and $\pi_j\in\Chcf_j(\pi_j\vee q\chi_i)$ for each $j\in J$ and $i\in I$.
Thus, $\pi$ is ex ante stable.

Since every worker is assigned with total probability one under $\pi$, every allocation in any decomposition of $\pi$ must assign all three workers.
In any feasible integral allocation where all three workers are assigned, firm $j_2$ cannot take $i_1$ because $j_1$ would then need to take $\{i_2,i_3\}$, which is infeasible.
Therefore, $\pi$ cannot be decomposed into feasible integral allocations.
\end{example}

\section{Ordinal Preference Reports}\label{sec:ordinal}
In our model, each agent has an M$^\natural$-concave valuation over deterministic bundles.
However, accurately eliciting cardinal information can be difficult in practice \citep{budish2022can}.
We consider ordinal reports over deterministic bundles, which do not uniquely determine agents' cardinal valuations.
We therefore examine the robustness of our results to this loss of information.
Specifically, we ask whether these reports provide enough information to identify stable and fair fractional allocations and select utility-preserving lottery implementations.

We first show that ordinal reports determine the fractional choice correspondences and hence the set of doubly-strong ex ante stable allocations.
Selecting a utility-preserving lottery may nevertheless require additional cardinal information.
We then identify a subclass of \emph{matroid-responsive} valuations for which ordinal information suffices for both steps.

\subsection{Identifying Stable Allocations from Ordinal Reports}\label{subsec:ordinal-identification}\label{subsec:ordinal-mnatural-market}
Consider the market of \Cref{sec:model}, with each agent $k$ having a true M$^\natural$-concave valuation $v_k$.
Agent $k$ reports the weak order $\succsim_k$ over $\{0,1,\ldots,q\}^{A_k}$ induced by her valuation:
\[
s\succsim_k t\quad\Longleftrightarrow\quad v_k(s)\ge v_k(t)
\qquad(s,t\in\{0,1,\ldots,q\}^{A_k}).
\]
Infeasible bundles have value $-\infty$ and are tied below every feasible bundle.
A valuation $u_k$ is \emph{compatible with the report} if it is M$^\natural$-concave, satisfies $\dom(u_k)\subseteq\{0,1,\ldots,q\}^{A_k}$ and $u_k(\bm{0})=0$, and induces the same weak order on $\{0,1,\ldots,q\}^{A_k}$.
The true valuation $v_k$ is compatible with the report, but the report does not uniquely determine it.

The report determines the agent's deterministic choice from every integral availability vector.
Our next theorem shows that this information also determines fractional choice.
For a finite coordinate set $E$ and a valuation $h\colon\ZZ^E\to\RRbar$ with $h(\bm{0})=0$ and $\dom(h)\subseteq\{0,1,\ldots,q\}^E$, write
\begin{align*}
\Chcd_h(s)&\coloneqq\argmax\{h(t)\mid \bm{0}\le t\le s,\ t\in\ZZ^E\},\\
\Chcf_h(x)&\coloneqq\argmax\{\overline h(z)\mid \bm{0}\le z\le x\}.
\end{align*}

\begin{theorem}\label{thm:choice-identification}
Let $v,u\colon\ZZ^E\to\RRbar$ be normalized M$^\natural$-concave functions with $\dom(v),\dom(u)\subseteq\{0,1,\ldots,q\}^{E}$.
Then, $\Chcd_v(s)=\Chcd_u(s)$ for every $s\in\{0,1,\ldots,q\}^{E}$ if and only if $\Chcf_v(x)=\Chcf_u(x)$ for every $x\in\RR_+^{E}$.
\end{theorem}

The proof uses \Cref{thm:Mintersection} to construct a joint distribution over integral selected and availability bundles that transfers fractional optimality from $v$ to $u$, as shown in Appendix~\ref{app:section3-proofs}.
The theorem requires only agreement of deterministic choices, which is weaker than agreement of the full weak orders over deterministic bundles.
Consequently, all valuations compatible with an agent's report induce the same fractional choice correspondence.\footnote{The M$^\natural$-concavity assumption is essential for this identification: a valuation outside this class can induce different fractional choices even when it represents the same weak order (\Cref{prop:rationalization-class-matters} in Appendix~\ref{app:rationalization-class}).}

Applying the tie-breaking rule of \Cref{subsec:derive-choice-functions} to these common correspondences gives the same induced choice functions and AG-preference orders.
We can therefore identify the stable and fair fractional allocations without eliciting the agents' cardinal valuations.

\begin{corollary}\label{cor:ordinal-invariance}
Fix a profile of ordinal reports induced by M$^\natural$-concave valuations.
All compatible valuation profiles induce the same set of doubly-strong ex ante stable allocations.
This set is nonempty and forms a distributive lattice under the common induced AG-preference orders on the two sides.
\end{corollary}
\begin{proof}
The reports determine each agent's deterministic choice correspondence.
By \Cref{thm:choice-identification}, the fractional choice correspondences are therefore common to all compatible valuation profiles.
The tie-breaking rule then yields the same induced choice functions and AG-preference orders.
By \Cref{thm:ag-ds-equivalence}, the doubly-strong ex ante stable set under each compatible profile coincides with this common AG-stable set.
Finally, \Cref{thm:cardinal-main} gives nonemptiness and the distributive lattice structure.
\end{proof}

Thus, any compatible valuation profile can be used to identify fractional allocations that are stable and fair under the agents' true valuations.
This conclusion concerns the information needed to identify these allocations; it does not require the utility levels assigned to them to be the same across compatible valuations.
Eliciting additional cardinal information beyond the ordinal reports therefore offers no benefit for identifying the set of doubly-strong ex ante stable allocations.

\subsection{Limits of Utility-Preserving Implementation}\label{subsec:nonresponsive-ambiguity}
The same reports also determine which deterministic allocations are stable, since deterministic stability depends only on the agents' choice correspondences.
Consequently, for a given fractional allocation, ordinal reports determine the set of lotteries that induce it and are supported on stable deterministic allocations.
However, they need not determine which of these lotteries preserve every agent's concave-closure utility.

For each compatible valuation profile, \Cref{thm:ex-ante-to-ex-post} provides a utility-preserving lottery over stable deterministic allocations for every doubly-strong ex ante stable fractional allocation.
The lottery may, however, depend on the cardinal information omitted from the reports.
The following example shows that no single implementing lottery need preserve utility under all compatible M$^\natural$-concave valuation profiles.

\begin{example}\label{ex:ordinal-ambiguous-order}
Let $I=\{i_1,i_2,i_3,i_4\}$, $J=\{s,t\}$, and $A=I\times J$.
Each worker has unit demand and is indifferent between the two firms, with $s\sim_i t\succ_i\varnothing$.
For each worker $i$, use the valuation $u_i(\varnothing)=0$, $u_i(\{s\})=u_i(\{t\})=1$, and $u_i(\{s,t\})=-\infty$.
Firm $t$ has capacity two and is indifferent among workers.
Its valuation is $u_t(X)=|X|$ for $|X|\le 2$ and $u_t(X)=-\infty$ otherwise.
Firm $s$ has no feasibility restrictions other than a capacity of two, and thus every bundle of size at most two is feasible.
Suppose the firm's weak order over feasible bundles is
\[
\begin{aligned}
\{i_1,i_2\}&\succ_s\{i_1,i_3\}\succ_s\{i_2,i_3\}\succ_s\{i_1,i_4\}\\
&\succ_s\{i_2,i_4\}\succ_s\{i_3,i_4\}\succ_s\{i_1\}\succ_s\{i_2\}\succ_s\{i_3\}\succ_s\{i_4\}\succ_s\varnothing.
\end{aligned}
\]
This weak order over feasible bundles is consistent with more than one M$^\natural$-concave valuation for $s$.
To see this, let $s$ have two positions $s_1$ and $s_2$, each compatible with all workers, and define a valuation by maximum weight matching as follows.\footnote{This is an assignment valuation in the sense of \citet{hatfield2005matching}. Such valuations are M$^\natural$-concave \citep[Section~3.6]{murota2016}.}
For any bundle $X$, let the value be the maximum weight of a matching between $X$ and $\{s_1,s_2\}$ that matches all elements of $X$, and set the value to $-\infty$ if no such matching exists.
Consider the following three weight matrices:
\[
\begin{array}{c|cccc}
 & i_1 & i_2 & i_3 & i_4\\ \hline
s_1 & 8 & 0 & 5 & 3\\
s_2 & 0 & 7 & 6 & 3
\end{array}
\qquad
\begin{array}{c|cccc}
 & i_1 & i_2 & i_3 & i_4\\ \hline
s_1 & 9 & 7 & 6 & 0\\
s_2 & 0 & 8 & 6 & 4
\end{array}
\qquad
\begin{array}{c|cccc}
 & i_1 & i_2 & i_3 & i_4\\ \hline
s_1 & 8 & 7 & 5 & 0\\
s_2 & 0 & 7 & 6 & 4
\end{array}
\]
Each matrix induces the same weak order on the feasible bundles listed above.
Let $v$, $v'$, and $v''$ denote the valuations obtained from the first, second, and third matrices, respectively.
All three valuations are M$^\natural$-concave because they are maximum weight matching valuations.

Let $r=(1/2,1/2,1/2,1/2)$ be firm $s$'s fractional bundle, and define the three decompositions
\begin{align*}
\lambda^{1}&=\tfrac{1}{2}\delta_{\{i_1,i_3\}}+\tfrac{1}{2}\delta_{\{i_2,i_4\}},&
\lambda^{2}&=\tfrac{1}{2}\delta_{\{i_1,i_2\}}+\tfrac{1}{2}\delta_{\{i_3,i_4\}},&
\lambda^{3}&=\tfrac{1}{2}\delta_{\{i_1,i_4\}}+\tfrac{1}{2}\delta_{\{i_2,i_3\}}.
\end{align*}
The set of optimal decompositions is $\conv\{\lambda^1,\lambda^2\}$ for $v$, $\conv\{\lambda^3,\lambda^2\}$ for $v'$, and $\conv\{\lambda^1,\lambda^3\}$ for $v''$.

These three sets have empty intersection.
Hence, no single decomposition is optimal for all three compatible valuations.

Define $\pi$ by $\pi_{is}=\pi_{it}=1/2$ for every $i\in I$.
Each optimal decomposition listed above lifts to a stable lottery inducing $\pi$ by assigning the selected pair to $s$ and its complement to $t$.
The allocation $\pi$ is doubly-strong ex ante stable under all three compatible valuations because every worker exhausts her unit demand, all assignment probabilities are equal, and both firms choose their full-capacity vectors.
Since every lottery inducing $\pi$ projects at firm $s$ to a decomposition of $r$, the empty intersection above implies that no single lottery is utility-preserving under all three compatible valuations.
\end{example}
Thus, ordinal reports identify fractional choice and the associated stable and fair fractional allocations, but selecting a utility-preserving stable implementation may require additional cardinal information.
A natural question is therefore whether suitable restrictions on valuations can eliminate the need for additional cardinal information at the implementation stage.

\subsection{Matroid-Responsive Valuations}\label{subsec:ordinal-model-stability}
We now restrict the agents' cardinal valuations to a subclass for which ordinal information also suffices for utility-preserving implementation.
These valuations are additive under matroid constraints, and thus expected utility depends only on the fractional allocation and not on its decomposition.
This subclass encompasses the models of \citet{kesten2015theory} and \citet{cookson2025fairly} and accommodates type-specific quotas \citep{abdulkadirouglu2003school}.

We specialize the model of \Cref{sec:model} to binary matches: each worker--firm pair can be matched at most once.
Let $A\subseteq I\times J$ be the set of mutually weakly acceptable pairs, with omitted pairs constrained to have zero allocation, and let $A_k$ be the pairs incident to agent $k$.
Since each pair appears only once, we identify the elements of $A_k$ with $k$'s potential partners.
Each agent $k$ has a family $\mathcal F_k\subseteq 2^{A_k}$ of feasible partner sets.

A nonempty family $\mathcal F_k$ is a \emph{matroid} if (i) $X\in\mathcal F_k$ and $Y\subseteq X$ imply $Y\in\mathcal F_k$, and (ii) whenever $X,Y\in\mathcal F_k$ with $|X|<|Y|$, there exists $a\in Y\setminus X$ such that $X\cup\{a\}\in\mathcal F_k$.\footnote{The indicator vectors of the sets in a matroid form an M$^\natural$-convex set (\Cref{def:Mnatural-convex-set}).}
The latter property permits a smaller feasible bundle to be expanded using an element of a larger one.
A capacity constraint is the special case $\mathcal F_k=\{X\subseteq A_k\mid |X|\le c_k\}$.
More generally, the class covers type-specific quotas \citep{abdulkadirouglu2003school}.
Formally, let $(A_{kt})_{t\in T_k}$ partition $A_k$ into types, each with a nonnegative integer upper quota $c_{kt}$.
Together with an overall capacity $c_k$, these quotas define $\mathcal F_k=\{X\subseteq A_k\mid |X|\le c_k,\ |X\cap A_{kt}|\le c_{kt}\text{ for all }t\in T_k\}$.
This family is a truncated partition matroid.

We assume that each $\mathcal F_k$ is a matroid.
Its rank function is
$r_k(S)=\max\{|X|\mid X\subseteq S,\ X\in\mathcal F_k\}$, and its polytope is
\[
P(r_k)=\left\{x\in\RR_+^{A_k}\ \middle|\ \sum_{a\in S}x_a\le r_k(S)\ \ (\forall S\subseteq A_k)\right\}.
\]
The equivalent rank-function formulation is reviewed in Appendix~\ref{subsec:valuations}.
For $X\subseteq A_k$, let $\chi_X$ be its indicator vector, so that $X\in\mathcal F_k$ if and only if $\chi_X\in P(r_k)$.

We call agent $k$'s valuation \emph{matroid-responsive} if, for some nonnegative weights $w_k=(w_{ka})_{a\in A_k}$, it takes the form
\[
v_k(Y)=
\begin{cases}
\sum_{a\in Y}w_{ka} & \text{if }Y\in\mathcal F_k,\\
-\infty & \text{otherwise}.
\end{cases}
\]
We assume that each agent's true valuation has this form.
The weights $w_k$ are cardinal primitives, but they need not be reported.
Given the matroid constraints, agent $k$ reports only the weak order $\succsim_k$ over $A_k\cup\{\varnothing\}$ induced by these weights, where the outside option $\varnothing$ has weight $w_{k\varnothing}=0$.
Thus, all partners under consideration are weakly preferred to the outside option.
For $a,a'\in A_k\cup\{\varnothing\}$, we write $a\succ_k a'$ if $a\succsim_k a'$ and $a'\not\succsim_k a$, and we write $a\sim_k a'$ if $a\succsim_k a'$ and $a'\succsim_k a$.
For a set $X\subseteq A$, write $X_k\coloneqq X\cap A_k$.
An allocation $X\subseteq A$ is feasible if each $X_k$ lies in $\mathcal{F}_k$.

The reported order and the matroid constraint define a deterministic choice correspondence $\Chcd_k\colon 2^{A_k}\rightrightarrows 2^{A_k}$ by selecting feasible subsets that admit no improvement by adding a partner strictly preferred to the outside option or replacing a selected partner with a higher-ranked one.
Formally,
\begin{align*}
\Chcd_k(X)\coloneqq\left\{Y\subseteq X \ \middle|\
\begin{array}{l}
Y\in\mathcal F_k,\\
\nexists a\in Y\cup\{\varnothing\},\ a'\in X\setminus Y\ \text{with}\\
\qquad a'\succ_k a,\ (Y\setminus\{a\})\cup\{a'\}\in\mathcal F_k
\end{array}
\right\}\qquad(\forall X\subseteq A_k).
\end{align*}

A weight vector $w$ on $A_k\cup\{\varnothing\}$ is \emph{consistent} with the report $\succsim_k$ if $w_\varnothing=0$, $w_a>w_{a'}$ whenever $a\succ_k a'$, and $w_a=w_{a'}$ whenever $a\sim_k a'$.
For any such weights $w$, let $v_k^w$ be obtained from $v_k$ by replacing $w_k$ with $w$; in particular, $v_k^{w_k}=v_k$.
We identify subsets with their indicator vectors and extend $v_k^w$ by $-\infty$ outside $\{0,1\}^{A_k}$.
The deterministic choice correspondence induced by $v_k^w$ is
\begin{align*}
\Chcd_k^{w}(X)\coloneqq\argmax\left\{\sum_{a\in Y} w_a \mid Y\subseteq X,\ Y\in \mathcal{F}_k\right\}\qquad(\forall X\subseteq A_k).
\end{align*}

\begin{proposition}\label{prop:responsive_unique}
For every weight vector $w$ consistent with $\succsim_k$ and every $X\subseteq A_k$, we have
\[
\Chcd_k^{w}(X)=\Chcd_k(X).
\]
\end{proposition}
\begin{proof}
All partner weights are nonnegative because every partner is weakly preferred to the outside option.
By the maximum-weight matroid greedy theorem \citep{edmonds1971matroids}, a feasible subset $Y$ of $X$ maximizes total weight if and only if there is no feasible addition or exchange that replaces some $a\in Y\cup\{\varnothing\}$ by a strictly higher-weight partner $a'\in X\setminus Y$.
Because $w$ is consistent with $\succsim_k$, this condition is exactly the local-improvement characterization of $\Chcd_k(X)$ above.
\end{proof}

Every $v_k^w$ is a normalized M$^\natural$-concave function, and \Cref{lem:matroid-polytope-integral} gives its concave closure as
\[
\overline{v_k^w}(x)=
\begin{cases}
\sum_{a\in A_k}w_ax_a & \text{if }x\in P(r_k),\\
-\infty & \text{otherwise}
\end{cases}
\qquad(\forall x\in[0,1]^{A_k}).
\]
By \Cref{prop:responsive_unique} and \Cref{thm:choice-identification}, all consistent weights induce the same deterministic and fractional choice correspondences, with the latter denoted by $\Chcf_k$.
Thus, partner rankings determine choice even though they need not determine a complete ranking of feasible bundles.

These valuations satisfy the assumptions of \Cref{sec:model} with $q=1$, and thus the stability and fairness definitions apply directly to the worker--firm pairs.
By \Cref{thm:ag-ds-equivalence}, doubly-strong ex ante stability coincides with AG-stability for the induced choice functions and depends only on the partner rankings and matroid constraints.
Together with \Cref{thm:cardinal-main,thm:ex-ante-to-ex-post}, this gives the following result for markets with matroid-responsive valuations.

\begin{corollary}\label{cor:matroid-responsive-implementation}
Suppose that each agent $k$ has a matroid-responsive valuation $v_k$ and reports the induced weak order $\succsim_k$ over partners and the outside option, with the matroid constraint $r_k$ given.
Then, the following hold.
\begin{enumerate}[(i)]
\item \emph{Stable fractional allocations.}
The set of doubly-strong ex ante stable fractional allocations is nonempty and forms a distributive lattice under the common induced AG-preference orders.
It depends only on these reports and constraints, equivalently on their induced profile of deterministic choice correspondences.

\item \emph{Utility preservation.}
For every such allocation $\pi$, every lottery $\lambda\in\Delta(\cM)$ inducing $\pi$ satisfies
\[
\sum_{X\in\cM}\lambda_Xv_k^{w^k}(X_k)=\overline{v_k^{w^k}}(\pi_k)
\qquad(\forall k\in I\cup J)
\]
for every consistent weight profile $w=(w^k)_{k\in I\cup J}$.

\item \emph{Stable lottery implementation.}
Every such allocation $\pi$ is induced by some lottery $\lambda\in\Delta(\cM)$.
Moreover, every lottery inducing $\pi$ is supported on deterministic allocations that are stable under every consistent weight profile.
\end{enumerate}
\end{corollary}
\begin{proof}
The preceding identification and \mbox{\Cref{thm:cardinal-main}} give a common nonempty stable set and distributive lattice for all consistent weight profiles.
Fix $\pi$ in this set.
By \Cref{prop:joint-implementability}, $\pi$ is induced by some lottery in $\Delta(\cM)$.
For any lottery $\lambda\in\Delta(\cM)$ inducing $\pi$, any consistent profile $w=(w^k)_{k\in I\cup J}$, and any agent $k$, additivity gives
\[
\sum_{X\in\cM}\lambda_X v_k^{w^k}(X_k)
=\sum_{a\in A_k}w^k_a\pi_a
=\overline{v_k^{w^k}}(\pi_k),
\]
and thus $\lambda$ is utility-preserving under every consistent profile, including the true one.
Since the stable set is common to all consistent profiles, $\pi$ is ex ante stable under each of them, and \Cref{thm:ex-ante-to-ex-post} implies that $\lambda$ is ex post stable under each consistent profile.
\end{proof}

\section{Conclusion}\label{sec:conclusion}
This paper studies ex ante stable and fair random allocations in two-sided markets with ties and nonresponsive bundle preferences.
We show that, when each agent has an M$^\natural$-concave valuation, a doubly-strong ex ante stable allocation exists and the set of such allocations forms a distributive lattice under the induced preference orders.
Every such allocation admits a lottery supported on stable deterministic allocations that simultaneously attains all agents' concave-closure utilities.

We further show that deterministic and fractional choice correspondences determine one another throughout the M$^\natural$-concave class.
Consequently, under M$^\natural$-concavity, ordinal reports over deterministic bundles identify the set of stable and fair fractional allocations and its lattice structure.
The same ordinal information also identifies all stable-support lottery implementations, but cardinal information can matter for identifying which of them preserve utility.

In the matroid-responsive subclass, ordinal reports and known matroid constraints suffice to select a utility-preserving stable lottery.
The subclass includes capacity constraints and type-specific quotas.
However, some applications fall outside this subclass, including certain type-specific reserve policies \citep{hafalir2013effective,ehlers2014school}.
Nevertheless, in school choice, students have unit demand, and thus every lottery implementing a given fractional allocation preserves their utilities under all cardinal representations consistent with their rankings.
Hence, when welfare is evaluated only on the student side, selecting a utility-preserving stable implementation requires no additional cardinal information, even when schools' M$^\natural$-concave valuations are not matroid-responsive.

There are two directions for future research.
First, the literature establishes impossibility results for jointly achieving stability, fairness, and strategy-proofness \citep{kesten2015theory,bando2025impossibility}.
These results motivate future research on the extent of manipulability of mechanisms that select stable and fair allocations.
Second, while the Alkan--Gale iteration converges, it need not terminate in finitely many steps \citep{AlkanGale2003,cookson2025fairly}.
For additive valuations with quotas, polynomial-time algorithms are available.
\citet{cookson2025fairly} develop one for one-to-one matching with indifferences on both sides, and \citet{karzanov2024mixed} develops one for capacitated many-to-many markets with weak orders.
Whether doubly-strong ex ante stable allocations can be computed in polynomial time for general M$^\natural$-concave valuations remains open.

\bibliographystyle{abbrvnat}
\bibliography{bibliography}

@article{bogomolnaia2001new,
  title={A new solution to the random assignment problem},
  author={Bogomolnaia, Anna and Moulin, Herv{\'e}},
  journal={Journal of Economic Theory},
  volume={100},
  number={2},
  pages={295--328},
  year={2001},
  publisher={Elsevier}
}

@article{cowgill2025stable,
  title={Stable matching on the job? {T}heory and evidence on internal talent markets},
  author={Cowgill, Bo and Davis, Jonathan M. V. and Montagnes, B. Pablo and Perkowski, Patryk},
  journal={Management Science},
  volume={71},
  number={3},
  pages={2508--2526},
  year={2025},
  publisher={INFORMS}
}

@article{he2018pseudo,
  title={A pseudo-market approach to allocation with priorities},
  author={He, Yinghua and Miralles, Antonio and Pycia, Marek and Yan, Jianye},
  journal={American Economic Journal: Microeconomics},
  volume={10},
  number={3},
  pages={272--314},
  year={2018},
  publisher={American Economic Association 2014 Broadway, Suite 305, Nashville, TN 37203-2425}
}

@article{echenique2021constrained,
  title={Constrained pseudo-market equilibrium},
  author={Echenique, Federico and Miralles, Antonio and Zhang, Jun},
  journal={American Economic Review},
  volume={111},
  number={11},
  pages={3699--3732},
  year={2021},
  publisher={American Economic Association 2014 Broadway, Suite 305, Nashville, TN 37203}
}

@article{katta2006solution,
  title={A solution to the random assignment problem on the full preference domain},
  author={Katta, Akshay-Kumar and Sethuraman, Jay},
  journal={Journal of Economic Theory},
  volume={131},
  number={1},
  pages={231--250},
  year={2006},
  publisher={Elsevier}
}

@article{caragiannis2019stable,
  title={Stable fractional matchings},
  author={Caragiannis, Ioannis and Filos-Ratsikas, Aris and Kanellopoulos, Panagiotis and Vaish, Rohit},
  journal={Artificial Intelligence},
  volume={295},
  pages={103416},
  year={2021},
  doi={10.1016/j.artint.2020.103416}
}

@article{bando2025impossibility,
  title={Impossibility results for weak strategy-proofness and respect for improvements in random assignment with priorities},
  author={Bando, Keisuke and Takase, Souta},
  journal={Economics Letters},
  volume={257},
  pages={112700},
  year={2025},
  publisher={Elsevier},
  doi={10.1016/j.econlet.2025.112700}
}

@article{roth1993stable,
  title={Stable matchings, optimal assignments, and linear programming},
  author={Roth, Alvin E. and Rothblum, Uriel G. and Vande Vate, John H.},
  journal={Mathematics of Operations Research},
  volume={18},
  number={4},
  pages={803--828},
  year={1993},
  publisher={INFORMS}
}

@article{budish2013designing,
  title={Designing random allocation mechanisms: Theory and applications},
  author={Budish, Eric and Che, Yeon-Koo and Kojima, Fuhito and Milgrom, Paul},
  journal={American Economic Review},
  volume={103},
  number={2},
  pages={585--623},
  year={2013},
  publisher={American Economic Association}
}

@article{hylland1979efficient,
  title={The efficient allocation of individuals to positions},
  author={Hylland, Aanund and Zeckhauser, Richard},
  journal={Journal of Political Economy},
  volume={87},
  number={2},
  pages={293--314},
  year={1979},
  publisher={The University of Chicago Press}
}

@article{han2024theory,
  title={A theory of fair random allocation under priorities},
  author={Han, Xiang},
  journal={Theoretical Economics},
  volume={19},
  number={3},
  pages={1185--1221},
  year={2024},
  publisher={Wiley Online Library}
}

@article{yilmaz2009random,
  title={Random assignment under weak preferences},
  author={Y{\i}lmaz, {\"O}zg{\"u}r},
  journal={Games and Economic Behavior},
  volume={66},
  number={1},
  pages={546--558},
  year={2009},
  publisher={Elsevier}
}

@article{budish2022can,
  title={Can market participants report their preferences accurately (enough)?},
  author={Budish, Eric and Kessler, Judd B.},
  journal={Management Science},
  volume={68},
  number={2},
  pages={1107--1130},
  year={2022},
  publisher={INFORMS}
}

@incollection{edmonds1970submodular,
  title={Submodular Functions, Matroids, and Certain Polyhedra},
  author={Edmonds, Jack},
  booktitle={Combinatorial Structures and Their Applications},
  editor={Guy, Richard and Hanani, Haim and Sauer, Norbert and Sch{\"o}nheim, J.},
  pages={69--87},
  year={1970},
  publisher={Gordon and Breach},
  address={New York}
}

@article{erdil2008whats,
  title={What's the matter with tie-breaking? {I}mproving efficiency in school choice},
  author={Erdil, Aytek and Ergin, Haluk},
  journal={American Economic Review},
  volume={98},
  number={3},
  pages={669--689},
  year={2008},
  publisher={American Economic Association}
}

@inproceedings{echenique2024diversity,
  title={Diversity in Choice as Majorization},
  author={Echenique, Federico and Mekonnen, Teddy and Yenmez, M. Bumin},
  booktitle={Proceedings of the 26th ACM Conference on Economics and Computation},
  pages={1025},
  year={2025},
  doi={10.1145/3736252.3742666},
  note={Full version: arXiv:2407.17589}
}

@article{abdulkadirouglu2009strategy,
  title={Strategy-proofness versus efficiency in matching with indifferences: Redesigning the {NYC} high school match},
  author={Abdulkadiro{\u{g}}lu, Atila and Pathak, Parag A. and Roth, Alvin E.},
  journal={American Economic Review},
  volume={99},
  number={5},
  pages={1954--1978},
  year={2009},
  publisher={American Economic Association}
}

@article{gul1999walrasian,
  title={Walrasian equilibrium with gross substitutes},
  author={Gul, Faruk and Stacchetti, Ennio},
  journal={Journal of Economic Theory},
  volume={87},
  number={1},
  pages={95--124},
  year={1999},
  publisher={Elsevier}
}

@article{yang2017maximal,
  title={On the maximal domain theorem: A corrigendum to ``{Walrasian} equilibrium with gross substitutes''},
  author={Yang, Yi-You},
  journal={Journal of Economic Theory},
  volume={172},
  pages={505--511},
  year={2017},
  doi={10.1016/j.jet.2017.09.003},
  publisher={Elsevier}
}

@article{fujishige2003note,
  title={A note on {K}elso and {C}rawford's gross substitutes condition},
  author={Fujishige, Satoru and Yang, Zaifu},
  journal={Mathematics of Operations Research},
  volume={28},
  number={3},
  pages={463--469},
  year={2003},
  publisher={INFORMS}
}

@article{kelso1982job,
  title={Job matching, coalition formation, and gross substitutes},
  author={Kelso, Alexander S. and Crawford, Vincent P.},
  journal={Econometrica},
  volume={50},
  number={6},
  pages={1483--1504},
  year={1982},
  publisher={Wiley-Blackwell}
}

@article{aygun2021college,
  title={College admission with multidimensional privileges: The {B}razilian affirmative action case},
  author={Ayg{\"u}n, Orhan and B{\'o}, In{\'a}cio},
  journal={American Economic Journal: Microeconomics},
  volume={13},
  number={3},
  pages={1--28},
  year={2021},
  publisher={American Economic Association 2014 Broadway, Suite 305, Nashville, TN 37203-2425}
}

@article{sonmez2022affirmative,
  title={Affirmative action in {I}ndia via vertical, horizontal, and overlapping reservations},
  author={S{\"o}nmez, Tayfun and Yenmez, M. Bumin},
  journal={Econometrica},
  volume={90},
  number={3},
  pages={1143--1176},
  year={2022},
  publisher={Wiley-Blackwell}
}

@article{abdulkadirouglu2003school,
  title={School choice: A mechanism design approach},
  author={Abdulkadiro{\u{g}}lu, Atila and S{\"o}nmez, Tayfun},
  journal={American Economic Review},
  volume={93},
  number={3},
  pages={729--747},
  year={2003}
}

@article{ehlers2014school,
  title={School choice with controlled choice constraints: Hard bounds versus soft bounds},
  author={Ehlers, Lars and Hafalir, Isa E. and Yenmez, M. Bumin and Yildirim, Muhammed A.},
  journal={Journal of Economic Theory},
  volume={153},
  pages={648--683},
  year={2014},
  publisher={Elsevier}
}

@article{hafalir2013effective,
  title={Effective affirmative action in school choice},
  author={Hafalir, Isa E. and Yenmez, M. Bumin and Yildirim, Muhammed A.},
  journal={Theoretical Economics},
  volume={8},
  number={2},
  pages={325--363},
  year={2013},
  publisher={Wiley Online Library}
}

@article{kojima2018designing,
  title={Designing matching mechanisms under constraints: An approach from discrete convex analysis},
  author={Kojima, Fuhito and Tamura, Akihisa and Yokoo, Makoto},
  journal={Journal of Economic Theory},
  volume={176},
  pages={803--833},
  year={2018},
  publisher={Elsevier}
}

@misc{okumura2026simple,
  title={A Simple Method for School Choice Lotteries},
  author={Okumura, Yasunori},
  year={2026},
  archivePrefix={arXiv},
  eprint={2605.06721},
  primaryClass={cs.GT},
  url={https://arxiv.org/abs/2605.06721}
}

@misc{aziz2026smart,
  title={Smart Lotteries in School Choice: Ex-ante {P}areto-Improvement with Ex-post Stability},
  author={Aziz, Haris and Bir{\'o}, P{\'e}ter and Cs{\'a}ji, Gergely and Demeulemeester, Tom},
  year={2026},
  archivePrefix={arXiv},
  eprint={2602.10679},
  primaryClass={cs.GT},
  url={https://arxiv.org/abs/2602.10679}
}

@article{murota2016,
  title={Discrete Convex Analysis: A Tool for Economics and Game Theory},
  author={Murota, Kazuo},
  journal={Journal of Mechanism and Institution Design},
  volume={1},
  number={1},
  pages={151--273},
  year={2016}
}

@BOOK{murota2003,
  title = {Discrete Convex Analysis},
  publisher = {SIAM},
  year = {2003},
  author = {Murota, Kazuo},
  address = {Philadelphia},
}

@article{aygun2013matching,
  title={Matching with contracts: Comment},
  author={Ayg{\"u}n, Orhan and S{\"o}nmez, Tayfun},
  journal={American Economic Review},
  volume={103},
  number={5},
  pages={2050--2051},
  year={2013}
}

@article{hatfield2005matching,
  title={Matching with contracts},
  author={Hatfield, John William and Milgrom, Paul R.},
  journal={American Economic Review},
  volume={95},
  number={4},
  pages={913--935},
  year={2005}
}

@article{IK2025,
  title = {Efficient and Strategy-proof Mechanism under General Constraints},
  author = {Imamura, Kenzo and Kawase, Yasushi},
  journal = {Theoretical Economics},
  volume = {20},
  number = {2},
  pages = {481--509},
  year = {2025},
  doi = {10.3982/TE6039}
}

@article{MurotaYokoi2015,
  title = {On the {{Lattice Structure}} of {{Stable Allocations}} in a {{Two-Sided Discrete-Concave Market}}},
  author = {Murota, Kazuo and Yokoi, Yu},
  year = {2015},
  journal = {Mathematics of Operations Research},
  volume = {40},
  number = {2},
  pages = {460--473},
}

@article{FujishigeTamura2006,
title = {A general two-sided matching market with discrete concave utility functions},
journal = {Discrete Applied Mathematics},
volume = {154},
number = {6},
pages = {950-970},
year = {2006},
author = {Satoru Fujishige and Akihisa Tamura},
}

@article{FujishigeTamura2007,
  title={A Two-Sided Discrete-Concave Market with Possibly Bounded Side Payments: An Approach by Discrete Convex Analysis},
  author={Fujishige, Satoru and Tamura, Akihisa},
  journal={Mathematics of Operations Research},
  volume={32},
  number={1},
  pages={136--155},
  year={2007},
  doi={10.1287/moor.1070.0227}
}

@article{Fujishige1980lex,
  author = {Fujishige, Satoru},
  title = {Lexicographically optimal base of a polymatroid with respect to a weight vector},
  journal = {Mathematics of Operations Research},
  volume = {5},
  number = {2},
  pages = {186--196},
  year = {1980},
}

@article{AlkanGale2003,
  title={Stable schedule matching under revealed preference},
  author={Alkan, Ahmet and Gale, David},
  journal={Journal of Economic Theory},
  volume={112},
  number={2},
  pages={289--306},
  year={2003},
}

@inproceedings{BIK2025,
  title={Properties of Path-Independent Choice Correspondences and Their Applications to Efficient and Stable Matchings},
  author={Keisuke Bando and Kenzo Imamura and Yasushi Kawase},
  booktitle={Proceedings of the 26th ACM Conference on Economics and Computation},
  publisher={Association for Computing Machinery},
  address={New York, NY, USA},
  pages={4},
  doi={10.1145/3736252.3742484},
  year={2025},
  note={Full version: \url{https://arxiv.org/abs/2502.09265}}
}

@article{MurotaShioura2000,
  title = {Extension of {{M-Convexity}} and {{L-Convexity}} to {{Polyhedral Convex Functions}}},
  author = {Murota, Kazuo and Shioura, Akiyoshi},
  year = 2000,
  journal = {Advances in Applied Mathematics},
  volume = {25},
  number = {4},
  pages = {352--427},
}

@inproceedings{cookson2025fairly,
  title={Fairly Stable Two-Sided Matching with Indifferences},
  author={Cookson, Benjamin and Shah, Nisarg},
  booktitle={Proceedings of the 26th ACM Conference on Economics and Computation},
  series={EC '25},
  publisher={Association for Computing Machinery},
  address={New York, NY, USA},
  pages={1130},
  note={Full version: \url{https://www.cs.toronto.edu/~nisarg/papers/dfda.pdf}},
  year={2025},
  doi={10.1145/3736252.3742675}
}

@article{karzanov2024mixed,
  title={On stable assignments generated by choice functions of mixed type},
  author={Karzanov, Alexander V.},
  journal={Discrete Applied Mathematics},
  volume={358},
  pages={112--135},
  year={2024},
  doi={10.1016/j.dam.2024.06.037},
}

@article{kesten2015theory,
  title={A theory of school-choice lotteries},
  author={Kesten, Onur and {\"U}nver, M. Utku},
  journal={Theoretical Economics},
  volume={10},
  number={2},
  pages={543--595},
  year={2015},
  publisher={Wiley Online Library}
}

@inproceedings{EguchiA2003Generalized,
  author       = {Eguchi, Akinobu and Fujishige, Satoru and Tamura, Akihisa},
  title        = {A Generalized {G}ale--{S}hapley Algorithm for a Discrete--Concave Stable--Marriage Model},
  booktitle    = {Proceedings of the 14th International Symposium on Algorithms and Computation (ISAAC 2003)},
  series       = {Lecture Notes in Computer Science},
  volume       = {2906},
  pages        = {495--504},
  publisher    = {Springer-Verlag},
  address      = {Berlin},
  year         = {2003},
}

@inproceedings{KawaseSumitaYokoi,
  title={Random Assignment of Indivisible Goods under Constraints},
  author={Kawase, Yasushi and Sumita, Hanna and Yokoi, Yu},
  booktitle={Proceedings of the Thirty-Second International Joint Conference on Artificial Intelligence},
  publisher={International Joint Conferences on Artificial Intelligence},
  address={Macao, SAR, China},
  pages={2792--2799},
  year={2023}
}

@article{FujishigeSanoZhan2018,
  title={The Random Assignment Problem with Submodular Constraints on Goods},
  author={Fujishige, Satoru and Sano, Yoshio and Zhan, Ping},
  journal={ACM Transactions on Economics and Computation},
  volume={6},
  number={1},
  pages={1--28},
  year={2018},
  doi={10.1145/3175496}
}

@article{edmonds1971matroids,
  author  = {Edmonds, Jack},
  title   = {Matroids and the Greedy Algorithm},
  journal = {Mathematical Programming},
  year    = {1971},
  volume  = {1},
  number  = {1},
  pages   = {127--136},
  doi     = {10.1007/BF01584082},
}

@article{murota1999m,
  title={M-convex function on generalized polymatroid},
  author={Murota, Kazuo and Shioura, Akiyoshi},
  journal={Mathematics of Operations Research},
  volume={24},
  number={1},
  pages={95--105},
  year={1999},
  publisher={INFORMS}
}

@article{roth1984evolution,
  title={The evolution of the labor market for medical interns and residents: a case study in game theory},
  author={Roth, Alvin E.},
  journal={Journal of Political Economy},
  volume={92},
  number={6},
  pages={991--1016},
  year={1984},
  publisher={The University of Chicago Press}
}

@article{erdil2014strategy,
  title={Strategy-proof stochastic assignment},
  author={Erdil, Aytek},
  journal={Journal of Economic Theory},
  volume={151},
  pages={146--162},
  year={2014},
  publisher={Elsevier}
}

@article{FrankTardos1988,
  title={Generalized polymatroids and submodular flows},
  author={Frank, Andr{\'a}s and Tardos, {\'E}va},
  journal={Mathematical Programming},
  volume={42},
  number={1--3},
  pages={489--563},
  year={1988}
}

@incollection{Frank1984c,
  title={Generalized polymatroids},
  author={Frank, Andr{\'a}s},
  booktitle={Finite and Infinite Sets},
  editor={Hajnal, Andr{\'a}s and Lov{\'a}sz, L{\'a}szl{\'o} and S{\'o}s, Vera T.},
  series={Colloquia Mathematica Societatis J{\'a}nos Bolyai},
  volume={37},
  pages={285--294},
  year={1984},
  publisher={North-Holland},
  address={Amsterdam, New York}
}

@article{aziz2022vigilant,
  title={The vigilant eating rule: A general approach for probabilistic economic design with constraints},
  author={Aziz, Haris and Brandl, Florian},
  journal={Games and Economic Behavior},
  volume={135},
  pages={168--187},
  year={2022}
}

@article{Atkinson1970inequality,
  author  = {Atkinson, Anthony B.},
  title   = {On the Measurement of Inequality},
  journal = {Journal of Economic Theory},
  year    = {1970},
  volume  = {2},
  number  = {3},
  pages   = {244--263},
  doi     = {10.1016/0022-0531(70)90039-6}
}

@article{DuttaRay1989egalitarian,
  author  = {Dutta, Bhaskar and Ray, Debraj},
  title   = {A Concept of Egalitarianism Under Participation Constraints},
  journal = {Econometrica},
  year    = {1989},
  volume  = {57},
  number  = {3},
  pages   = {615--635}
}

@article{milgrom2009substitute,
  author  = {Milgrom, Paul and Strulovici, Bruno},
  title   = {Substitute Goods, Auctions, and Equilibrium},
  journal = {Journal of Economic Theory},
  year    = {2009},
  volume  = {144},
  number  = {1},
  pages   = {212--247}
}

@article{shioura2015gross,
  author  = {Shioura, Akiyoshi and Tamura, Akihisa},
  title   = {Gross Substitutes Condition and Discrete Concavity for Multi-Unit Valuations: A Survey},
  journal = {Journal of the Operations Research Society of Japan},
  year    = {2015},
  volume  = {58},
  number  = {1},
  pages   = {61--103}
}

\appendix

\section{Discrete-Convex Tools}\label{app:discrete-convex-tools}\label{sec:preliminaries}
This appendix collects the discrete-convex and polyhedral results used in the proofs.
The discrete exchange condition of \Cref{def:Mnatural-concavity} is understood on any finite coordinate set, not only an individual agent's set of partners.
For a function $g\colon X\to\RRbar$, write $\dom(g)=\{x\in X\mid g(x)>-\infty\}$.
For a finite set $A$, a vector $x\in\RR^A$, and $S\subseteq A$, write $x(S)=\sum_{a\in S}x_a$ and let $\chi_S$ be the indicator vector of $S$.
We use $\supp(x)=\{a\mid x_a\ne0\}$, $\supp^+(x)=\{a\mid x_a>0\}$, and $\supp^-(x)=\{a\mid x_a<0\}$.

\subsection{Feasible Sets and Polyhedral Representations}\label{subsec:valuations}
We begin with standard set-function notions.
A set function $r\colon 2^A\to\RR$ is \emph{submodular} if it satisfies
\begin{align*}
r(S)+r(T)\ge r(S\cup T)+r(S\cap T) \quad \text{for all } S,T\subseteq A.
\end{align*}
We call $p\colon 2^A\to\RR$ \emph{supermodular} if $-p$ is submodular.
We say that $r$ is \emph{monotone} if $r(S)\le r(T)$ for all $S\subseteq T\subseteq A$.
For a monotone submodular function $r$ with $r(\emptyset)=0$, the \emph{polymatroid} and \emph{base polyhedron} associated to $r$ are the polytopes
\begin{align*}
P(r)&\coloneqq\{x\in\RR_+^A\mid x(S)\le r(S)\ \ (\forall S\subseteq A)\},\\
B(r)&\coloneqq\{x\in\RR_+^A\mid x(S)\le r(S)\ \ (\forall S\subseteq A),\ x(A)=r(A)\}.
\end{align*}
The polymatroid $P(r)$ and base polyhedron $B(r)$ are called \emph{integral} if $r$ is integer-valued.
We use two equivalent descriptions of a matroid.
A \emph{matroid rank function} is an integer-valued monotone submodular function $r\colon 2^A\to\ZZ_+$ satisfying $r(S)\le |S|$ for all $S\subseteq A$.
An \emph{independence system} on $A$ is a nonempty family $\mathcal{F}\subseteq 2^A$ such that $X\subseteq Y\in\mathcal{F}$ implies $X\in\mathcal{F}$.
A \emph{matroid} on $A$ is an independence system $\mathcal{F}$ with the exchange property that for any $X,Y\in\mathcal{F}$ with $|X|<|Y|$, there exists $y\in Y\setminus X$ such that $X\cup\{y\}\in\mathcal{F}$.
The two views coincide.
From $r$ we obtain $\{S\subseteq A\mid r(S)=|S|\}$, and from $\mathcal{F}$ we obtain $r(S)\coloneqq \max\{|T|\mid T\subseteq S,\ T\in\mathcal{F}\}$.

We next define generalized polymatroids.
A \emph{generalized polymatroid (g-polymatroid, or M$^\natural$-convex polyhedron)} is a polyhedron of the form
\begin{align*}
P(\rho,\mu)\coloneqq\{x\in\RR^A\mid \mu(S)\le x(S)\le \rho(S)\ \ (\forall S\subseteq A)\},
\end{align*}
where $\rho$ is submodular, $\mu$ is supermodular, $\rho(\emptyset)=\mu(\emptyset)=0$, and $\rho(S)-\mu(T)\ge \rho(S\setminus T)-\mu(T\setminus S)$ for all $S,T\subseteq A$.
If $\rho$ and $\mu$ are integer-valued then $P(\rho,\mu)$ is called integral.

The next lemma states a standard integrality property of polymatroids and g-polymatroids.
\begin{lemma}[\citet{edmonds1971matroids} and \citet{murota2003}]\label{lem:matroid-polytope-integral}
If $P\subseteq \RR^A$ is an integral polymatroid or an integral g-polymatroid, then $P=\conv\bigl(P\cap\ZZ^A\bigr)$.
\end{lemma}

Recall the exchange definition of an M$^\natural$-convex set in \Cref{def:Mnatural-convex-set}.
Integral g-polymatroids provide a geometric representation of M$^\natural$-convex sets.
If $P\subseteq\RR^A$ is an integral g-polymatroid, then $P\cap\ZZ^A$ is M$^\natural$-convex.
Conversely, every M$^\natural$-convex set can be represented as the integer points of an integral g-polymatroid.
In particular, for a matroid rank function $r$, the set $P(r)\cap\ZZ_+^A$ is an M$^\natural$-convex set.

The effective domain of any M$^\natural$-concave function is an M$^\natural$-convex set.
A typical example of an M$^\natural$-concave function is a modular function on an M$^\natural$-convex set.
Let $Q\subseteq\ZZ^A$ be M$^\natural$-convex and let $w\in\RR^A$.
Define $f\colon \ZZ^A\to\RRbar$ by $f(x)=\sum_{a\in A} w_a x_a$ if $x\in Q$ and $f(x)=-\infty$ otherwise.
This $f$ is M$^\natural$-concave.
As another example, let $U$ be a finite set and consider a weighted bipartite graph with bipartition $(A,U)$ and edge weights $w_{a,u}$.
For $x\in\{0,1\}^A$, let $X=\{a\in A \mid x_a=1\}$ and define $g(x)$ as the maximum weight of a matching between $X$ and $U$ that covers every vertex in $X$.
Set $g(x)=-\infty$ if $x\notin\{0,1\}^A$ or no such matching exists.
This maximum weight matching valuation is M$^\natural$-concave.

A function $g\colon\ZZ^A\to\RRbar$ with a bounded effective domain is M$^\natural$-concave if and only if, for every vector $p\in\RR^A$, $\argmax\{g(x)-\textstyle\sum_{a\in A}p_ax_a \mid x\in\ZZ^A\}$ is an M$^\natural$-convex set~\citep{murota1999m}.
Moreover, if $A$ is a disjoint union $A_1\cup\cdots\cup A_n$ and each $Q_\ell\subseteq\ZZ^{A_\ell}$ is M$^\natural$-convex, then $\prod_{\ell=1}^n Q_\ell$ is M$^\natural$-convex.

For more details on M$^\natural$-concavity and related concepts, we refer to \citet{murota2003,murota2016}.

\subsection{Concave Closures and Continuous Exchange}\label{subsec:relaxation}
Let $g\colon\ZZ^A\to\RRbar$ have a nonempty finite effective domain.
As in \Cref{subsec:fractional-allocations}, its concave closure $\overline g(z)$ is the maximum expected $g$-value over lotteries on $\dom(g)$ with mean $z$, and equals $-\infty$ outside $\conv(\dom(g))$.

We also use the continuous analogue of M$^\natural$-concavity, which we define below.
A function $h\colon\RR^A\to\RRbar$ is called \emph{polyhedral concave} if its hypograph $\{(x,p)\mid x\in\RR^A,\,h(x)\ge p\}$ forms a polyhedron.
A polyhedral concave function $h\colon \RR^A\to\RRbar$ is said to be M$^\natural$-concave if, for all $x,y\in\dom(h)$ and $i\in\supp^+(x-y)$, there exist $j\in\supp^-(x-y)\cup\{\varnothing\}$ and $\eta_0>0$ such that
\begin{align*}
h(x)+h(y)\le h(x-\eta(\chi_i-\chi_j))+h(y+\eta(\chi_i-\chi_j))&&(\forall \eta\in[0,\eta_0]).
\end{align*}

The following theorem characterizes M$^\natural$-concave functions in terms of g-polymatroids.
\begin{theorem}[\cite{MurotaShioura2000,murota2003}]\label{thm:Mconcave-g} 
\leavevmode\par\noindent
A polyhedral concave function $h\colon \RR^A\to\RRbar$ is M$^\natural$-concave if and only if, for every vector $p\in\RR^A$, the set
\[
\argmax\left\{h(x)-\sum_{a\in A}p_ax_a\mid x\in\RR^A\right\}
\]
is a g-polymatroid whenever it is nonempty.
Specifically, if $h\colon \RR^A\to\RRbar$ is the concave closure of an M$^\natural$-concave function on $\ZZ^A$, then this maximizer set is an integral g-polymatroid whenever it is nonempty.
\end{theorem}

We will also use the following exchange-based criterion, which characterizes global maximizers by local directional derivatives.
\begin{theorem}[\cite{MurotaShioura2000,murota2003}]\label{thm:optimality}
For a polyhedral M$^\natural$-concave function $h\colon \RR^A\to\RRbar$ and $x\in\dom(h)$, we have $h(x)\ge h(y)$ for all $y\in\RR^A$ if and only if
\begin{align*}
h'(x;-\chi_i+\chi_j)\le 0 \quad(\forall i,j\in A\cup\{\varnothing\}),
\end{align*}
where $h'(x;d)\coloneqq \lim_{\alpha\searrow 0}\frac{h(x+\alpha d)-h(x)}{\alpha}$ denotes the one-sided directional derivative of $h$ at $x\in\dom(h)$ in direction $d\in\RR^A$.
\end{theorem}

The \emph{integral neighborhood} of $\pi\in\RR^A$ is
\begin{align*}
N(\pi)\coloneqq\left\{x\in\ZZ^A\mid \lfloor\pi_a\rfloor\le x_a\le\lceil\pi_a\rceil\ \ (\forall a\in A)\right\}.
\end{align*}
For the function $g$ above, write $N_g(\pi)\coloneqq N(\pi)\cap\dom(g)$.
Its \emph{local concave closure} $\tilde{g} \colon \RR^A \to \RRbar$ is defined, for each $\pi\in\RR^A$, by
\begin{align*}
\tilde{g}(\pi)\coloneqq\max\left\{\sum_{x\in N_g(\pi)}\lambda_{x}g(x) \ \mid
\sum_{x\in N_g(\pi)}\lambda_x x=\pi,
\sum_{x\in N_g(\pi)}\lambda_x=1,\
\lambda_x\ge 0~(\forall x\in N_g(\pi))\right\}.
\end{align*}
We set $\tilde g(\pi)=-\infty$ if no such convex combination exists.

The next theorem shows that global and local closures coincide for M$^\natural$-concave functions, which allows us to work with local neighborhoods while retaining M$^\natural$-concavity.
\begin{theorem}[\cite{MurotaShioura2000,murota2003}]\label{thm:local-global-closure}
Let $g\colon\ZZ^A\to\RRbar$ be M$^\natural$-concave with finite effective domain.
Its concave closure $\overline g$ and local concave closure $\tilde g$ coincide, and their common extension is polyhedral M$^\natural$-concave.
\end{theorem}

Restricting a polyhedral M$^\natural$-concave function to a coordinate box, and setting it to $-\infty$ outside that box, preserves polyhedral M$^\natural$-concavity.
Indeed, the restriction preserves polyhedrality, and the exchange step can be chosen sufficiently small to keep both exchanged points within the box.
This fact justifies the lower-box restrictions used for fractional choice and local optimality.

\section{Omitted Proofs}\label{app:section3-proofs}

\begin{proof}[Proof of \Cref{prop:lorenz-choice}]
Fix $x\in\RR_+^A$ and set $C=\Chcf(x)$, a nonempty compact generalized polymatroid contained in $[0,q]^A$.
The maximum-total face $B$ of $C$ is a base polyhedron \citep[Proposition~1.11]{FrankTardos1988}.
By \citet{Fujishige1980lex}, $B$ contains a unique Lorenz-dominant fractional bundle $b$.

We next show that $b$ Lorenz-dominates every fractional bundle in $C$.
For any $y\in C$, choose $z$ to maximize total assignment over $\{w\in C\mid w\ge y\}$.
If $z\notin B$, the local optimality criterion in \Cref{thm:optimality}, applied to the modular objective $\|w\|_1$ on $C$ and $-\infty$ elsewhere, gives a feasible coordinate increase, contradicting the choice of $z$.
Hence, $z\in B$ and $z\ge y$, and $b$ therefore Lorenz-dominates $y$.
Any Lorenz-dominant fractional bundle in $C$ must belong to $B$ by the comparison of total quantities, and thus $b$ is unique.

Lorenz dominance and uniqueness imply that $b$ uniquely maximizes $y^\uparrow$ in lexicographic order over $C$.
On the fixed-total set $B$, $d(y)$ differs from $\sum_{a\in A}y_a^2$ by a constant, and thus $b$ minimizes $d$ over $B$ \citep{Fujishige1980lex}.
Since $d$ is strictly decreasing in each coordinate on $[0,q]^A$, the extension $z\ge y$ above gives $d(b)\le d(z)\le d(y)$ for every $y\in C$.
Strict convexity of $d$ ensures that $b$ is the unique minimizer over $C$.
\end{proof}

\begin{proof}[Proof of \Cref{lem:induced-choice-regularity}]
We first prove consistency.
Suppose that $\Chf(x)\le y\le x$.
Then, $\Chf(x)$ remains feasible and utility-maximizing under cap $y$, and every maximizer under cap $y$ is also a maximizer under cap $x$.
Thus,
\[
\Chcf(y)=\Chcf(x)\cap\{z\mid z\le y\}.
\]
The set on the right contains the unique $d$-minimizer $\Chf(x)$, and its unique $d$-minimizer is therefore also $\Chf(x)$.
Hence, $\Chf(y)=\Chf(x)$.

\medskip
Next, we prove persistence.
Suppose, for contradiction, that there exist $x, y \in \RR_+^A$ with $x \ge y$ but $\Chf(y) \not\ge \Chf(x) \wedge y$.
Let $x' = \Chf(x)$ and $y' = \Chf(y)$.
Since $y' \not\ge x' \wedge y$, there exists some $a \in A$ with $y'_a < \min\{x'_a, y_a\}$.
Therefore, $a \in \supp^+(x' - y')$.

By the exchange axiom, there exist $a' \in \supp^-(x' - y') \cup \{\varnothing\}$ and $\eta_0 > 0$ such that
\begin{align}
\overline{g}(x') + \overline{g}(y') \leq \overline{g}(x' - \eta(\chi_a - \chi_{a'})) + \overline{g}(y' + \eta(\chi_a - \chi_{a'})) \qquad (\forall \eta \in [0, \eta_0]). \label{eq:persist_exchange}
\end{align}
If $a'=\varnothing$, set $\eta'=\min\{\eta_0,x'_a,y_a-y'_a\}$.
If $a'\in A$, set $\eta'=\min\{\eta_0,x'_a,y_a-y'_a,x_{a'}-x'_{a'},y'_{a'}\}$.
The minima are positive because $y'_a<y_a$, $0\le y'_a<x'_a$ and, when $a'\in A$, we have $0\le x'_{a'}<y'_{a'}\le y_{a'}\le x_{a'}$.
Then, $\bm{0}\le x' - \eta'(\chi_a - \chi_{a'}) \le x$ and $\bm{0}\le y' + \eta'(\chi_a - \chi_{a'}) \le y$.
Because the exchange inequality has finite left-hand side, it also ensures that both perturbed points lie in $\dom(\overline g)$.
They are therefore feasible for the same argmax problems, and comparing $d$ at these points is valid.
Since $x' \in \argmax_{z \le x} \overline{g}(z)$, we have $\overline{g}(x' - \eta'(\chi_a - \chi_{a'})) \le \overline{g}(x')$.
Similarly, $y' \in \argmax_{z \le y} \overline{g}(z)$ implies $\overline{g}(y' + \eta'(\chi_a - \chi_{a'})) \le \overline{g}(y')$.
These inequalities, together with \eqref{eq:persist_exchange}, yield $\overline{g}(x' - \eta'(\chi_a - \chi_{a'})) = \overline{g}(x')$ and $\overline{g}(y' + \eta'(\chi_a - \chi_{a'})) = \overline{g}(y')$.
Moreover, concavity of $\overline{g}$ ensures $\overline{g}(x' - \eta(\chi_a - \chi_{a'})) = \overline{g}(x')$ and $\overline{g}(y' + \eta(\chi_a - \chi_{a'})) = \overline{g}(y')$ for all $\eta \in [0, \eta']$.

If $a'=\varnothing$, then $d(y' + \eta' \chi_a) < d(y')$, contradicting $y' = \Chf(y)$.
Thus, assume $a' \in A$.
The exchange move preserves $\sum_a y_a$, and $d(y)=\sum_{a\in A}y_a^2-2q\sum_{a\in A}y_a+|A|q^2$, hence comparing $d$ reduces to comparing $\sum_{a\in A}y_a^2$.
Since $d(x') \le d(x' - \eta(\chi_a - \chi_{a'}))$ for any $\eta \in [0, \eta']$, it follows that $x'_a \le x'_{a'}$.
In the same way, since $d(y') \le d(y' + \eta(\chi_a - \chi_{a'}))$ for any $\eta \in [0,\eta']$, we must have $y'_a \ge y'_{a'}$.
But this leads to $x'_a > y'_a \ge y'_{a'} > x'_{a'} \ge x'_a$, which is a contradiction.

Thus, $\Chf$ is persistent.

\medskip
Next, we prove size monotonicity.
Suppose for contradiction that there exist $x, y \in \RR_+^A$ with $x \ge y$ but $\|\Chf(x)\|_1<\|\Chf(y)\|_1$.
Let $x'=\Chf(x)$ and $y'=\Chf(y)$.
Then, $\|x'\|_1<\|y'\|_1$.
Introduce a new index $0$ and define $\hat{A}\coloneqq A\cup\{0\}$ and
\[
\hat{g}(z_0,z)\coloneqq
\begin{cases}
\overline{g}(z) & \text{if } z_0=-\|z\|_1,\\
-\infty & \text{otherwise}.
\end{cases}
\]
This is the standard dummy-coordinate lift of an M$^\natural$-concave function to an M-concave function \citep{murota2003}.
More explicitly, the dummy coordinate records the negative of the total quantity, which makes the coordinate sum of every lifted vector zero.
An exchange move of $\overline{g}$ remains an exchange between the corresponding ordinary coordinates, while a move with no exchange coordinate is paired with coordinate $0$.
Conversely, an exchange involving coordinate $0$ projects to a no-exchange move of $\overline{g}$.
The affine fixed-sum lift preserves polyhedrality and the exchange inequality, and therefore $\hat{g}$ satisfies the M-concave exchange axiom on its effective domain.
Let $\hat{x}=(-\|x'\|_1,x')$ and $\hat{y}=(-\|y'\|_1,y')$.
Since $\|x'\|_1<\|y'\|_1$, we have $\hat{x}_0>\hat{y}_0$ and $0\in\supp^+(\hat{x}-\hat{y})$.
Applying the exchange axiom with index $0$, there exist $a\in A$ with $x'_a<y'_a~(\le y_a\le x_a)$ and $\eta_0>0$ such that for every $\eta\in[0,\eta_0]$,
\begin{align*}
\hat{g}(\hat{x})+\hat{g}(\hat{y})
&\le \hat{g}(\hat{x}-\eta(\chi_0-\chi_a))+\hat{g}(\hat{y}+\eta(\chi_0-\chi_a)).
\end{align*}
This inequality is equivalent to
\[
\overline{g}(x')+\overline{g}(y')\le \overline{g}(x'+\eta\chi_a)+\overline{g}(y'-\eta\chi_a).
\]
Choose $\eta=\min\{\eta_0,\ x_a-x'_a,\ y'_a\}~(>0)$.
The exchange inequality ensures that both perturbed points lie in $\dom(\overline g)$, and the choice of $\eta$ gives $\bm{0}\le x'+\eta\chi_a\le x$ and $\bm{0}\le y'-\eta\chi_a\le y$.
Since $y'$ maximizes $\overline{g}$ over $\{z\mid z\le y\}$, we have $\overline{g}(y'-\eta\chi_a)\le \overline{g}(y')$.
Therefore, $\overline{g}(x'+\eta\chi_a)\ge \overline{g}(x')$.
Because $x'$ maximizes $\overline{g}$ over $\{z\mid z\le x\}$, this implies $x'+\eta\chi_a\in\Chcf(x)$.
Since $x'$ Lorenz-dominates every fractional bundle in $\Chcf(x)$, the comparison at $r=|A|$ implies that $x'$ maximizes total assignment over this set.
This contradicts $\|x'+\eta\chi_a\|_1>\|x'\|_1$.

Thus, $\Chf$ satisfies size monotonicity.

\medskip
It remains to prove continuity.
We establish the stronger global bound
\[
\|\Chf(x)-\Chf(y)\|_1\le 2\|x-y\|_1
\qquad(\forall x,y\in\RR_+^A).
\]
First suppose that $x\ge y$, and write $u=\Chf(x)$ and $v=\Chf(y)$.
Set
\[
P\coloneqq\sum_{a\in A}(u_a-v_a)_+,
\qquad
N\coloneqq\sum_{a\in A}(v_a-u_a)_+.
\]
If $u_a>v_a$, persistence gives $v_a\ge\min\{u_a,y_a\}$.
The minimum cannot be $u_a$.
Since $v_a\le y_a$, it follows that $v_a=y_a$.
Therefore, $u_a-v_a=u_a-y_a\le x_a-y_a$.
Consequently,
\[
P\le\sum_{a\in A}(x_a-y_a)=\|x-y\|_1.
\]
Size monotonicity gives $\|u\|_1\ge\|v\|_1$, equivalently $P\ge N$.
It follows that
\[
\|\Chf(x)-\Chf(y)\|_1=P+N
\le 2\|x-y\|_1
\qquad(x\ge y).
\]
For arbitrary $x,y$, let $w=x\vee y$.
Applying the comparable-case bound to $w\ge x$ and $w\ge y$ and using the triangle inequality yields
\begin{align*}
\|\Chf(x)-\Chf(y)\|_1
&\le 2\|w-x\|_1+2\|w-y\|_1
=2\|x-y\|_1.
\end{align*}
Thus, $\Chf$ is globally $2$-Lipschitz in the $\ell_1$ norm, and in particular it is continuous.
\end{proof}

\begin{proof}[Proof of \Cref{lem:non-satiated}]
Since $\zeta=\Chf(\zeta)$, $\zeta'=\Chf(\zeta\vee q\chi_a)$, and $\zeta\ne\zeta'$, we have $\overline{g}(\zeta)\le\overline{g}(\zeta')$.
Because $\zeta'\le \zeta\vee q\chi_a$, only coordinate $a$ can increase.
Suppose that $\zeta'_a\le\zeta_a$.
Because $\zeta'\le\zeta$, the maximality of $\zeta$ under cap $\zeta$ gives $\overline{g}(\zeta')\le\overline{g}(\zeta)$.
The two values are therefore equal, and both points maximize $\overline{g}$ under the expanded cap.
Since $\zeta'$ is the $d$-minimizing expanded-cap maximizer, we obtain $d(\zeta')\le d(\zeta)$.
On the other hand, $\zeta'\le\zeta\le q\bm{1}$ and $\zeta'\ne\zeta$, while $d$ is strictly decreasing in each coordinate on $[0,q]^A$; hence $d(\zeta')>d(\zeta)$, a contradiction.
Thus, $\zeta'_a>\zeta_a$, and $\zeta'_{a'}\le\zeta_{a'}$ for every $a'\in A\setminus\{a\}$.
By M$^\natural$-concavity of $\overline{g}$, there exist $a'\in\supp^-(\zeta'-\zeta)\cup\{\varnothing\}$ and a positive number $\eta_0>0$ such that
\begin{align*}
    \overline{g}(\zeta')+\overline{g}(\zeta)
    \le \overline{g}(\zeta'-\eta(\chi_a-\chi_{a'}))+\overline{g}(\zeta+\eta(\chi_a-\chi_{a'}))
    \quad(\forall \eta\in[0,\eta_0]).
\end{align*}
Let $\eta_0'=\min\{\eta_0,\zeta_{a'}-\zeta'_{a'}\}$ if $a'\in A$ and $\eta_0'=\eta_0$ if $a'=\varnothing$.
We have $\eta_0'>0$: when $a'\in A$, this follows from $a'\in\supp^-(\zeta'-\zeta)$.
For $\eta\in[0,\eta_0']$, the exchange inequality ensures domain membership and nonnegativity of both perturbed points, while the definition of $\eta_0'$ gives $\zeta'-\eta(\chi_a-\chi_{a'})\le\zeta\vee q\chi_a$ and $\zeta+\eta(\chi_a-\chi_{a'})\le\zeta\vee q\chi_a$.
Since $\zeta'$ maximizes $\overline{g}$ over $\{y\mid y\le \zeta\vee q\chi_a\}$, we have $\overline{g}(\zeta')\ge\overline{g}(\zeta'-\eta(\chi_a-\chi_{a'}))$.
In particular, we have
\begin{align*}
\overline{g}(\zeta')+\overline{g}(\zeta)
\le \overline{g}(\zeta'-\eta(\chi_a-\chi_{a'}))+\overline{g}(\zeta+\eta(\chi_a-\chi_{a'}))
\le \overline{g}(\zeta')+\overline{g}(\zeta+\eta(\chi_a-\chi_{a'})),
\end{align*}
and hence $\overline{g}(\zeta)\le \overline{g}(\zeta+\eta(\chi_a-\chi_{a'}))$.

Suppose that $\overline{g}(\zeta)<\overline{g}(\zeta+\eta(\chi_a-\chi_{a'}))$ for some $\eta\in(0,\eta_0']$.
Then, by the concavity of $\overline{g}$, we have $\overline{g}(\zeta)<\overline{g}(\zeta+\eta'(\chi_a-\chi_{a'}))$ for all $\eta'\in(0,\eta]$.
Taking $\epsilon_0=\eta$ yields alternative~(i).

Otherwise, suppose that $\overline{g}(\zeta)=\overline{g}(\zeta+\eta(\chi_a-\chi_{a'}))$ for all $\eta\in(0,\eta_0']$.
Then, the exchange inequality implies $\overline{g}(\zeta')\le \overline{g}(\zeta'-\eta(\chi_a-\chi_{a'}))$ for all $\eta\in[0,\eta_0']$.
Since $\zeta'$ is a maximizer over the expanded region and $\zeta'-\eta(\chi_a-\chi_{a'})$ is feasible for that region, the reverse inequality also holds.
Thus, we have $\overline{g}(\zeta')=\overline{g}(\zeta'-\eta(\chi_a-\chi_{a'}))$ for all $\eta\in[0,\eta_0']$.
If $a'\in A$ and $\zeta'_{a'}< \zeta'_a$, we have $d(\zeta'-\eta(\chi_a-\chi_{a'}))<d(\zeta')$ for any $\eta\in (0,\min\{\eta_0',(\zeta'_a-\zeta'_{a'})/2\})$, which contradicts $\zeta'=\Chf(\zeta\vee q\chi_a)$.
Thus, $\zeta_{a'}\ge \zeta'_{a'}\ge \zeta'_a>\zeta_a$ if $a'\in A$.
Taking $\epsilon_0=\eta_0'$ yields alternative~(ii).
\end{proof}

\begin{proof}[Proof of \Cref{thm:ex-ante-to-ex-post}]
By \Cref{prop:joint-implementability}, a lottery that induces $\pi$ and attains every agent's concave-closure utility exists.
Fix any such lottery $\lambda$.
Suppose, for the sake of contradiction, that $\lambda$ is not ex post stable.
Then, there is an unstable deterministic allocation $x^*\in \supp(\lambda)$, i.e., $\lambda_{x^*}>0$.
\begin{itemize}
    \item Suppose that $x^*_k\not\in\Chcd_k(x^*_k)$ for some $k\in I\cup J$.
    Then, there exists $y^*\in\ZZ^{A_k}$ such that $v_k(y^*)>v_k(x^*_k)$ and $y^*\le x^*_k$.
    Let $\xi\coloneqq\pi_k-\lambda_{x^*}(x_k^*-y^*)$.
    In agent $k$'s induced lottery over bundles, replace the bundle $x_k^*$ at the realization $x^*$ by $y^*$ and leave the other realizations unchanged.
    It induces the fractional bundle $\xi$.
    By the convex-combination characterization of $\overline{v}_k$, the expected value under this lottery gives a lower bound on $\overline{v}_k(\xi)$.
    Then, we have $\xi\le \pi_k$ and
    \begin{align*}
        \overline{v}_k(\xi)
        &\ge \sum_{x\ne x^*}\lambda_x v_k(x_k)+\lambda_{x^*} v_k(y^*)
        =\overline{v}_k(\pi_k)+\lambda_{x^*}(v_k(y^*)-v_k(x_k^*))
        >\overline{v}_k(\pi_k).
    \end{align*}
    This means that $\pi$ is not ex ante stable, a contradiction.

    \item Suppose that there is a pair $(i,j)\in I\times J$ such that $x^*_i\not\in\Chcd_i(x^*_i\vee q\chi_j)$ and $x^*_j\not\in\Chcd_j(x^*_j\vee q\chi_i)$.
    Choose $y_i\le x^*_i\vee q\chi_j$ and $y_j\le x^*_j\vee q\chi_i$ with $v_i(y_i)>v_i(x^*_i)$ and $v_j(y_j)>v_j(x^*_j)$.
    Let $\lambda^*\coloneqq\lambda_{x^*}$ and define $\xi_i\coloneqq \pi_i+\lambda^*(y_i-x^*_i)$ and $\xi_j\coloneqq \pi_j+\lambda^*(y_j-x^*_j)$.
    We verify the required bounds for $\xi_i$; the argument for $\xi_j$ is symmetric.
    Since
    \[
    \xi_i=\sum_{x\ne x^*}\lambda_xx_i+\lambda^*y_i,
    \]
    we have $\xi_i\ge\bm{0}$.
    For every coordinate other than $(i,j)$, the inequality $y_i\le x^*_i\vee q\chi_j$ gives $(\xi_i)_a\le(\pi_i)_a$.
    At coordinate $(i,j)$, using $x_{ij}\le q$ for every $x\in\cM$ and $(y_i)_j\le q$, we obtain
    \[
    \pi_{ij}-\lambda^*x^*_{ij}
    =\sum_{x\ne x^*}\lambda_xx_{ij}
    \le(1-\lambda^*)q,
    \qquad
    (\xi_i)_j\le(1-\lambda^*)q+\lambda^*q=q.
    \]
    Thus $\bm{0}\le\xi_i\le\pi_i\vee q\chi_j$, and symmetrically $\bm{0}\le\xi_j\le\pi_j\vee q\chi_i$.
    Then, by the same argument as above, we have
    \begin{align*}
    \overline{v}_i(\xi_i)&\ge \overline{v}_i(\pi_i)+\lambda_{x^*}(v_i(y_i)-v_i(x^*_i))>\overline{v}_i(\pi_i),\\
    \overline{v}_j(\xi_j)&\ge \overline{v}_j(\pi_j)+\lambda_{x^*}(v_j(y_j)-v_j(x^*_j))>\overline{v}_j(\pi_j).
    \end{align*}
    Therefore $\pi_i\not\in\Chcf_i(\pi_i\vee q\chi_j)$ and $\pi_j\not\in\Chcf_j(\pi_j\vee q\chi_i)$.
    Thus, $\pi$ is not ex ante stable, a contradiction.
\end{itemize}
Therefore, $\lambda$ is ex post stable.
\end{proof}

\begin{proof}[Proof of \Cref{thm:choice-identification}]
Fix a finite set $E$ and $q\in\ZZ_{++}$, and let $Q\coloneqq\{0,1,\ldots,q\}^E$.
Fix normalized M$^\natural$-concave functions $v,u\colon\ZZ^E\to\RRbar$ with $\dom(v),\dom(u)\subseteq Q$.
We first prove that equality of their deterministic choice correspondences implies equality of their fractional choice correspondences.
The converse follows by restricting fractional choice to integral availability vectors.

Let $f\colon\ZZ^E\to\RRbar$ have a nonempty effective domain contained in $Q$.
For $z\in\conv(\dom(f))$, define
\[
\Lambda_f(z)\coloneqq
\left\{
\lambda\in\Delta(Q)
\ \middle|\
\supp(\lambda)\subseteq\dom(f),\
\sum_{t\in Q}\lambda_t t=z,\
\sum_{t\in\dom(f)}\lambda_t f(t)=\overline f(z)
\right\}.
\]
Thus, $\Lambda_f(z)$ is the set of utility-maximizing lotteries with mean $z$.

We begin with the common-decomposition property.

\begin{lemma}\label{lem:common-optimal-decomposition}
Let $f,g\colon\ZZ^E\to\RRbar$ be M$^\natural$-concave functions with nonempty effective domains contained in $Q$.
For every $z\in\conv(\dom(f))\cap\conv(\dom(g))$, we have $\Lambda_f(z)\cap\Lambda_g(z)\ne\varnothing$.
\end{lemma}

\begin{proof}
By concave-closure duality \citep{murota2003}, choose $(p,\alpha)\in\RR^E\times\RR$ such that $\alpha+p\cdot t\ge f(t)$ for every $t\in\dom(f)$ and $\overline f(z)=\alpha+p\cdot z$, and define
$D_f(p)\coloneqq\argmax_{t\in\dom(f)}\{f(t)-p\cdot t\}$.
An optimal lottery with mean $z$ exists because $\dom(f)$ is finite.
Since $\alpha+p\cdot t-f(t)\ge 0$ and its average under this lottery equals $\alpha+p\cdot z-\overline f(z)=0$, the lottery is supported on $D_f(p)$.
Hence, $z\in\conv(D_f(p))$.
Because $f(t)-p\cdot t$ is constant on $D_f(p)$, we also have $f(t)=\alpha+p\cdot t$ for every $t\in D_f(p)$.
Similarly, there are $(p',\alpha')\in\RR^E\times\RR$ such that $z\in\conv(D_g(p'))$ and $g(t)=\alpha'+p'\cdot t$ for every $t\in D_g(p')$.

Because subtracting a linear function preserves M$^\natural$-concavity, both maximizer sets are M$^\natural$-convex.
Since $z\in\conv(D_f(p))\cap\conv(D_g(p'))$, \Cref{thm:Mintersection} gives $z\in\conv(D_f(p)\cap D_g(p'))$.
Any convex decomposition of $z$ supported on this intersection is optimal for both $f$ and $g$, and hence belongs to $\Lambda_f(z)\cap\Lambda_g(z)$.
\end{proof}

The next lemma couples two optimal integral decompositions while preserving componentwise order.

\begin{lemma}\label{lem:choice-refinement}
Let $f,g\colon\ZZ^E\to\RRbar$ be normalized M$^\natural$-concave functions with $\dom(f),\dom(g)\subseteq Q$.
For every $\bm{0}\le z\le x\le q\bm{1}$ with $z\in\conv(\dom(f))$ and $x\in\conv(\dom(g))$, there is a lottery $\lambda$ on pairs $(p,s)\in Q\times Q$ such that $p\in\dom(f)$, $s\in\dom(g)$, and $p\le s$ for all $(p,s)\in Q\times Q$ with $\lambda_{(p,s)}>0$, $\sum_{(p,s)\in Q\times Q}\lambda_{(p,s)}p=z$, $\sum_{(p,s)\in Q\times Q}\lambda_{(p,s)}s=x$, $\sum_{(p,s)\in Q\times Q}\lambda_{(p,s)}f(p)=\overline f(z)$, and $\sum_{(p,s)\in Q\times Q}\lambda_{(p,s)}g(s)=\overline g(x)$.
\end{lemma}

\begin{proof}
Define $F,G\colon\ZZ^E\times\ZZ^E\to\RRbar$ by
\[
F(p,w)=
\begin{cases}
f(p) & \text{if }p\in Q,\ w\in Q,\\
-\infty & \text{otherwise},
\end{cases}
\qquad
G(p,w)=
\begin{cases}
g(p+w) & \text{if }p,w\ge\bm{0},\ p+w\in Q,\\
-\infty & \text{otherwise}.
\end{cases}
\]
The function $F$ is the direct sum of $f$ and the function that is zero on $Q$ and $-\infty$ elsewhere, and thus it is M$^\natural$-concave.
The function $G$ is obtained from $g$ by splitting each coordinate into two nonnegative coordinates, an operation that also preserves M$^\natural$-concavity \citep[Section~6.2]{murota2016}.
Both effective domains lie in $Q\times Q$.
In particular, the constraint $p+w\le q\bm{1}$ is imposed only on $G$; no such coupled restriction is needed for $F$.

We next show that
\[
\overline F(z,x-z)=\overline f(z)
\quad\text{and}\quad
\overline G(z,x-z)=\overline g(x).
\]
For $F$, any lottery with mean $(z,x-z)$ has expected value at most $\overline f(z)$.
Conversely, take an optimal lottery for $f$ with mean $z$ and independently draw $w\in Q$ from any lottery with mean $x-z$, which exists since $x-z\in[0,q]^E$.
This product lottery attains $\overline f(z)$.
For $G$, any lottery with mean $(z,x-z)$ induces a lottery over $s=p+w$ with mean $x$, and thus its expected value is at most $\overline g(x)$.
Conversely, draw $s$ from an optimal lottery for $g$ with mean $x$.
Conditional on $s$, for each coordinate $e$ with $x_e>0$, draw $p_e$ from a binomial distribution with $s_e$ trials and success probability $z_e/x_e$; set $p_e=0$ if $x_e=0$.
Then, $\bm{0}\le p\le s$, $\mathbb E[p]=z$, and, setting $w=s-p$, we have $\mathbb E[w]=x-z$ and $G(p,w)=g(s)$ at every support point.
This attains $\overline g(x)$.

Applying \Cref{lem:common-optimal-decomposition} on the doubled ground set, $F$ and $G$ admit a common optimal lottery with mean $(z,x-z)$.
Replacing each support point $(p,w)$ by $(p,s)=(p,p+w)$ yields the required lottery.
\end{proof}

We are now ready to prove \Cref{thm:choice-identification}.
Suppose first that $\Chcd_v(s)=\Chcd_u(s)$ for every $s\in Q$.
Define the monotone closure $u^+\colon\ZZ^E\to\RRbar$ by
\[
u^+(s)\coloneqq\max_{\bm{0}\le t\le s}u(t)\quad(s\in Q),
\]
and set $u^+(s)=-\infty$ outside $Q$.
The function $u^+$ is normalized and M$^\natural$-concave because supremal convolution and restriction to an integer rectangle preserve M$^\natural$-concavity \citep[see, e.g., Theorem~6.1 and Section~6.2]{murota2016}.
Moreover,
\begin{equation}
\overline{u^+}(x)=\max_{\bm{0}\le y\le x}\overline u(y)
\qquad(x\in[0,q]^E).
\label{eq:monotone-closure}
\end{equation}
Indeed, a lottery over $s$ with mean $x$, together with a maximizing bundle $t\le s$ at each support point, induces a lottery over $t$ with mean at most $x$.
Conversely, any lottery over bundles with mean at most $x$ can be augmented to a lottery over availability bundles with mean $x$.
Fix $\bm{0}\le z\le x\le q\bm{1}$ and suppose that $z\in\Chcf_v(x)$.
Apply \Cref{lem:choice-refinement} to $f=v$ and $g=u^+$, and let $\lambda$ be the resulting lottery.
For each $s\in Q$, choose $t^s\in\Chcd_v(s)$.
Let $\hat y\coloneqq\sum_{(p,s)\in Q\times Q}\lambda_{(p,s)}t^s$.
Since $t^s\le s$, \Cref{lem:choice-refinement} gives $\hat y\le\sum_{(p,s)\in Q\times Q}\lambda_{(p,s)}s=x$.
Because the lottery over the bundles $t^s$ has mean $\hat y$, its expected $v$-utility is at most $\overline v(\hat y)$ by the definition of the concave closure.
Because $p\le s$ and $t^s\in\Chcd_v(s)$, we have $v(p)\le v(t^s)$ for every $(p,s)\in\supp(\lambda)$.
Since $t^s\in\Chcd_v(s)$, $z\in\Chcf_v(x)$, and $\sum_{(p,s)\in Q\times Q}\lambda_{(p,s)}v(p)=\overline v(z)$ by \Cref{lem:choice-refinement}, we obtain
\begin{equation}
\sum_{(p,s)\in Q\times Q}\lambda_{(p,s)}v(p)\le\sum_{(p,s)\in Q\times Q}\lambda_{(p,s)}v(t^s)\le\overline v(\hat y)\le\max_{\bm{0}\le y\le x}\overline v(y)=\overline v(z)=\sum_{(p,s)\in Q\times Q}\lambda_{(p,s)}v(p).
\label{eq:choice-identification-v}
\end{equation}
Thus, all inequalities in \eqref{eq:choice-identification-v} hold with equality.
Since $v(t^s)\ge v(p)$ whenever $\lambda_{(p,s)}>0$, equality in the first inequality implies $p\in\Chcd_v(s)=\Chcd_u(s)$ for every $(p,s)\in\supp(\lambda)$.
Because the bundles $p$ have mean $z$, their expected $u$-utility is at most $\overline u(z)$.
By \Cref{lem:choice-refinement} and \eqref{eq:monotone-closure}, we then have
\[
\max_{\bm{0}\le y\le x}\overline u(y)
=\overline{u^+}(x)
=\sum_{(p,s)\in Q\times Q}\lambda_{(p,s)}u^+(s)
=\sum_{(p,s)\in Q\times Q}\lambda_{(p,s)}u(p)
\le\overline u(z)
\le\max_{\bm{0}\le y\le x}\overline u(y).
\]
Thus, $\overline u(z)=\max_{\bm{0}\le y\le x}\overline u(y)$, and hence $z\in\Chcf_u(x)$.
The reverse implication follows symmetrically.
Therefore, $\Chcf_v(x)=\Chcf_u(x)$ for every $x\in[0,q]^E$.
For any $x\in\RR_+^E$, we have $\Chcf_h(x)=\Chcf_h(x\wedge q\bm{1})$ for $h\in\{u,v\}$ because both effective domains lie in $Q$.
This proves the forward implication.

Conversely, suppose that $\Chcf_v(x)=\Chcf_u(x)$ for every $x\in\RR_+^E$.
For every $s\in Q$, \Cref{thm:local-global-closure} implies that the concave closure at any $y\le s$ is attained by a lottery over integral bundles bounded by $s$, and that $\overline h(t)=h(t)$ at every integral bundle $t$.
Hence, $\Chcd_h(s)=\Chcf_h(s)\cap\ZZ^E$ for $h\in\{u,v\}$.
Equality of fractional choices at $s$ therefore implies $\Chcd_v(s)=\Chcd_u(s)$, proving the converse.
\end{proof}

\section{The Role of the Valuation Domain}\label{app:rationalization-class}
The ordinal identification result in \Cref{subsec:ordinal-identification} maintains the assumption that valuations are M$^\natural$-concave.
The following example shows why this restriction matters, even when the feasible domain is fixed.
Here a \emph{rationalization} of a weak order over $2^E$ is a normalized valuation on that domain that induces the given order.

\begin{proposition}
\label{prop:rationalization-class-matters}
There exist a weak order $\succsim$ on $2^E$, an M$^\natural$-concave rationalization $v$ of $\succsim$, a rationalization $u$ of $\succsim$ that is not M$^\natural$-concave, and an availability vector $x\in[0,1]^E$ such that
\[
\Chcf_v(x)\ne\Chcf_u(x).
\]
\end{proposition}
\begin{proof}
Let $E=\{a,b,c\}$ and define
\[
v(X)
\coloneqq
4\cdot\mathbf{1}_{\{a\in X\}}
+2\cdot\mathbf{1}_{\{b\in X\}}
+\mathbf{1}_{\{c\in X\}}
-8\max\{|X|-2,0\}.
\]
The last term is a concave function of cardinality, and hence $v$ is M$^\natural$-concave.
Define $u$ by the following table, which also lists the values of $v$:
\[
\begin{array}{c|rrrrrrrr}
X
& \{a,b,c\} & \varnothing & \{c\} & \{b\}
& \{b,c\} & \{a\} & \{a,c\} & \{a,b\}\\
\hline
v(X)&-1&0&1&2&3&4&5&6\\
u(X)&-1&0&1&2&3&4&5&103
\end{array}
\]
Thus, $u$ and $v$ represent the same strict order.

Consider the availability vector $x=(1/2,1/2,1)$, with coordinates ordered as $(x_a,x_b,x_c)$.
Every lottery that induces a vector $z\le x$ satisfies $\mathbb{E}[v(X)]\le 4z_a+2z_b+z_c\le 4$.
The lottery $\tfrac12\delta_{\{a,c\}}+\tfrac12\delta_{\{b,c\}}$ induces $x$ and has expected $v$-value $4$.
Since all three coefficients in the second inequality are positive, equality requires $z=x$.
Therefore
\[
\Chcf_v(x)=\{(1/2,1/2,1)\}.
\]

For every $X\subseteq E$, direct inspection gives $u(X)\le 1+100\cdot\mathbf{1}_{\{a\in X\}}+2\cdot\mathbf{1}_{\{b\in X\}}$.
Consequently, every lottery that induces a vector $z\le x$ has expected $u$-value at most $1+100z_a+2z_b\le 52$.
The lottery $\tfrac12\delta_{\{a,b\}}+\tfrac12\delta_{\{c\}}$ attains this bound and induces $(1/2,1/2,1/2)$.
The only bundles attaining equality in the pointwise bound are $\{c\}$, $\{b,c\}$, and $\{a,b\}$.
Equality in the two positive-price cap constraints then forces an optimal lottery to assign probability $1/2$ to $\{a,b\}$, zero to $\{b,c\}$, and probability $1/2$ to $\{c\}$.
Thus,
\[
\Chcf_u(x)=\{(1/2,1/2,1/2)\}.
\]
The function $u$ is not M$^\natural$-concave.
Indeed, for $X=\{a,b\}$, $Y=\{c\}$, and the element $a\in X\setminus Y$, both permitted right-hand sides of the exchange inequality are $7$, whereas $u(X)+u(Y)=104$.
\end{proof}

\section{Incentives}\label{app:incentives}
Our existence results do not imply incentive compatibility.
We briefly discuss the relevant impossibilities, which already arise in responsive subdomains with unit-demand workers, including school-choice markets.

Consider a mechanism that maps reported ordinal preferences to fractional allocations.
It is \emph{SD-strategy-proof} if, for every worker and every misreport, the assignment obtained under truthful reporting first-order stochastically dominates the assignment obtained under the misreport.
It is \emph{weakly SD-strategy-proof} if no misreport yields an assignment that strictly first-order stochastically dominates the truthful assignment.

\citet{kesten2015theory} show that, with at least three students, no mechanism is both SD-strategy-proof and strongly ex ante stable.
In their school-choice domain, strong ex ante stability coincides with our stable and fair requirement by \Cref{prop:responsive-special-cases}.
Thus, even on the responsive school-choice subdomain in which each school is subject only to a capacity constraint, no mechanism that always selects a stable and fair allocation can be SD-strategy-proof.

\citet{bando2025impossibility} establish a stronger impossibility for weak SD-strategy-proofness.
They show that ex ante stability and weak equal-top fairness are incompatible with weak SD-strategy-proofness.
Their random-assignment model with priorities is a responsive subdomain of our model, and our stable and fair requirement implies weak equal-top fairness there.
Therefore, even on this responsive subdomain, no mechanism that always selects a stable and fair allocation can be weakly SD-strategy-proof.

An analogous conclusion holds under cardinal reports.
Cardinal strategy-proofness requires that no misreport increase a worker's expected utility evaluated at her true von Neumann--Morgenstern utility.
Represent each strict ranking in the impossibility instance of \citet{bando2025impossibility} by a fixed consistent unit-demand valuation.
A manipulation that yields a strictly first-order stochastically dominating assignment then strictly increases expected utility under the true valuation.
Hence, even on the corresponding responsive cardinal subdomain, no mechanism that always selects a stable and fair allocation can be cardinally strategy-proof.
In particular, selecting an outcome generated by the Alkan--Gale iteration does not restore strategy-proofness in general.

\end{document}